\documentclass[9pt,twocolumn,twoside]{osajnl}

\journal{josaa} 

\setboolean{shortarticle}{false}

\usepackage{amsmath,amssymb,amsthm}
\usepackage{mathrsfs,mathtools}

\newtheorem{lemma}{Lemma}
\newtheorem{theorem}{Theorem}
\newtheorem{definition}{Definition}
\newtheorem{corollary}{Corollary}
\newtheorem{remark}{Remark}

\newtheorem{proposition}{Proposition}

\title{On the optimal measurement of conversion gain in the presence of dark noise}

\author[1,*]{Aaron Hendrickson} 
\author[2]{David P. Haefner} 
\author[2]{Bradley L. Preece} 

\affil[1]{U.S. Navy, NAWCAD DAiTA Group, Atlantic Ranges \& Targets, Optical Systems, Electro-Optical Tracking Systems, 23013 Cedar Point Road, Bldg. 2118, Patuxent River MD, 20670}
\affil[2]{U.S. Army Combat Capabilities Development Command (DEVCOM), C5ISR Center, Research \& Technology Integration (RTI), 10221 Burbeck Road, Fort Belvoir, VA 22060}
\affil[*]{Corresponding author: ajh4184@gmail.com}

\begin{abstract}
Working from a model of Gaussian pixel noise, we present and unify over twenty-five years of developments in the statistical analysis of the photon transfer conversion gain measurement. We then study a two-sample estimator of the conversion gain that accounts for the general case of non-negligible dark noise. The moments of this estimator are ill-defined (their integral representations diverge) and so we propose a method for assigning pseudomoments, which are shown to agree with actual sample moments under mild conditions. A definition of optimal sample size pairs for this two-sample estimator is proposed and used to find approximate optimal sample size pairs that allow experimenters to achieve a predetermined measurement uncertainty with as little data as possible. The conditions under which these approximations hold are also discussed. Design and control of experiment procedures are developed and used to optimally estimate a per-pixel conversion gain map of a real image sensor.  Experimental results show excellent agreement with theoretical predictions and are backed up with Monte Carlo simulation. The per-pixel conversion gain estimates are then applied in a demonstration of per-pixel read noise estimation of the same image sensor. The results of this work open the door to a comprehensive pixel-level adaptation of the photon transfer method.
\end{abstract}

\setboolean{displaycopyright}{true}

\begin{document}

\maketitle


\section{Introduction}

Photon Transfer ({\sc pt}) is a methodology developed in the 1970s to aid in the design, characterization, and optimization of solid state image sensors \cite{janesick_2007}. Since its inception, {\sc pt} has evolved to become the standard approach to image sensor characterization for manufacturers and consumers alike, culminating in its use as the basis for the European Machine Vision Association ({\sc emva}) 1288 standard in 2005 \cite{EMVA_1288,EMVA_1288_4_linear}. To fully characterize the performance of an image sensor with {\sc pt}, many performance parameters are measured including, but not limited to, conversion gain, read noise, dynamic range, and quantum efficiency.

Of all performance parameters prescribed by the {\sc pt} method, the conversion gain, $g\,(e\text{-}/\mathrm{DN})$, is the most critical as it is the unit conversion constant needed to convert sensor measurements from device specific units of Digital Numbers ({\sc dn}) into physical units of electrons $(e\text{-})$. Since units of $\mathrm{DN}$ are device specific, it is only after multiplying by $g$ that $\mathrm{DN}$ measurements represent a physical quantity that can be compared between different devices. For this reason, many of the performance parameters measurable by the {\sc pt} method, e.g. read noise, dynamic range, and quantum efficiency, at some point require multiplying quantities in units of $\mathrm{DN}$ by $g$.

Naturally, $g$ must be estimated through measurement; thus, the precision and accuracy of its measurement fundamentally limits the precision and accuracy of all measured {\sc pt} parameters converted to units of electrons by $g$. To see why this is, suppose $G$ is an estimator of $g$ and $T$ is an estimator of some parameter $\tau$ with units of $\mathrm{DN}$. Then $\mathscr T=T\times G$ is an estimator of $\tau$ in units of electrons and we have for the absolute coefficient of variation $\mathsf{ACV}\mathscr T=\sqrt{\mathsf{Var}\mathscr T}/|\mathsf E\mathscr T|$:
\begin{equation}
\mathsf{ACV}^2\mathscr T=\mathsf{ACV}^2T+(\mathsf{ACV}^2T)(\mathsf{ACV}^2G)+\mathsf{ACV}^2G.
\end{equation}
It follows that $\mathsf{ACV}\mathscr T\geq\mathsf{ACV}G$ showing that the relative uncertainty in the estimator $G$ represents a lower bound for the relative uncertainty of all {\sc pt} measurements converted to units of electrons.

Given the central role conversion gain plays in the {\sc pt} method, much work has been conducted into investigating various estimators for $g$, their statistical properties, and procedures for performing $g$-estimation \cite{Beecken:96,EMVA_1288,EMVA_1288_4_linear,janesick_2001,pain_2003,janesick_2006,janesick_2007,bohndiek:2008,starkey_2016,Hendrickson:17,hendrickson_2019,hendrickson_2021,preece_22,Nakamoto_2022}. From this body of research, only a small subset have studied how sample size(s) relate to estimator uncertainty for the purpose of properly designing and controlling conversion gain measurement experiments \cite{Beecken:96,janesick_2007,Hendrickson:17,hendrickson_2021}. Perhaps of most influence to this correspondence is the work of Beecken \& Fossum (1996) who studied a shot-noise-limited estimator under a model of Gaussian noise \cite{Beecken:96}. When the shot-noise-limited assumption is made under the Gaussian model, statistical analysis of the resulting estimator becomes much more tractable and parameters such as the optimal (minimal) sample size needed to control the estimator's uncertainty can be derived. Of course, results derived from this assumption only hold in the limiting case of shot-noise-limited estimation and is generally not applicable for measuring $g$ in the presence of dark noise which is the more general situation. Such a situation can occur when dealing with sensors of lower-quality or when attempting to measure conversion gain under low-illumination conditions.

An important consideration in how conversion gain is measured depends on if a particular sensor under test exhibits uniform characteristics, whereby, each pixel in the sensor array is assumed to exhibit identical values of conversion gain, read noise, etc. When a sensor conforms to the uniform assumption, only a single, global, value of each {\sc pt} parameter needs to be reported to describe the performance of the device.  Since most sensors are comprised of many pixels, only a few frames of data are needed to obtain large sample estimates of these parameters. For example, Janesick showed that when working in the shot-noise-limited regime, one can measure the conversion gain to a relative uncertainty of $1\%$ by sampling $\mathcal O(10^4)$ (approximately $20\,000$) pixels from just a few frames of image data \cite{janesick_2007}. When the assumption of sensor uniformity is not held, it is necessary to measure each {\sc pt} performance parameter on a per-pixel basis, requiring many frames of data. Returning back to the example given by Janesick, the $\mathcal O(10^4)$ pixels needed to obtain a global estimate of $g$ within $1\%$ uncertainty now turns into $\mathcal O(10^4)$ full-resolution image frames needed to measure the conversion gain of each pixel to the same uncertainty.

The requirement of large datasets for per-pixel conversion gain estimation is further compounded in the case of Complementary Metal-Oxide Semiconductor ({\sc cmos}) Active Pixel Sensors ({\sc aps}), which not only exhibit per-pixel nonuniformities, but also are generally nonlinear devices. To circumvent the problem of nonlinearity, Janesick \emph{et al}.~proposed an extension of {\sc pt} which requires measuring $g$ at the low-illumination end of it's dynamic range, where the sensor exhibits linear characteristics (see Section $7.3$ of \cite{janesick_2006} and \cite{bohndiek:2008}). As sensor dark noise dominates signal noise at low-illumination, the number of samples needed to measure $g$ again dramatically increases. For example, in this correspondence we show that if photon induced signal noise is approximately equal to sensor dark noise, then to measure $g$, per-pixel, to $1\%$ uncertainty one needs no less than $\mathcal O(10^5)$ (approximately $180\, 000$) total image frames of data. The dramatic increase in the amount of data needed for per-pixel estimation, especially when measuring $g$ outside the shot-noise-limited regime, drives the need to find a means for estimating $g$ with as little data as possible; this is the main goal of this correspondence.

From an experimental perspective, capturing large datasets is not just time consuming but also introduces sensitivity to drift in the sensor and/or light source, which unchecked, will corrupt per-pixel estimates of $g$. Through utilizing optimal sampling, one is able to measure $g$ with as few frames as possible, reducing the time to capture data and mitigating drift. As such, our goal here is to develop a general method for optimally estimating $g$ that holds in the shot-noise-limit as well as the more general case where dark noise is non-negligible.

We will organize this paper by first introducing a Gaussian model of sensor noise along with the basic assumptions under which the model holds (Section \ref{sec:sensor_noise_model}). We follow with a condensed review of theory and statistical analysis for estimators of the conversion gain in the shot-noise-limited case (Section \ref{sec:shot_noise_limited}) and general case which accounts for the presence of dark noise (Section \ref{sec:sub_shot_noise_limited}). Section \ref{sec:sub_shot_noise_limited_optimal_sample sizes} then studies optimal sample size pairs for the general conversion gain estimator, which is an estimator based on two independent samples. Because analytical expressions for the exact optimal sample size pairs cannot be derived, this section derives approximations and determines the conditions under which these approximations are useful. From these analytical expressions for the approximate optimal sample size pairs, Section \ref{sec:DOE_COE} describes a method for design and control of experiment for per-pixel conversion gain estimation on a real image sensor. Lastly, Section \ref{sec:R_map} presents an application to per-pixel conversion gain estimation by performing per-pixel estimation of read noise on the same sensor.


\section{Sensor noise model}
\label{sec:sensor_noise_model}

Consider a sensor observing a constant irradiance, monochromatic light source. The expected number of interacting photons, $\mu_\gamma$, received by each pixel per fixed integration time is given by
\begin{equation}
\mu_\gamma=\frac{AEt_\mathrm{exp}}{h\nu}QE_\mathrm{int},
\end{equation}
where $A\, (\mathrm m^2)$ is the pixel area, $E\, (\mathrm W/\mathrm m^2)$ is the irradiance, $t_\mathrm{exp}\, (\mathrm s)$ is the integration time, $h\nu\, (\mathrm J)$ is the quantization constant for photons of frequency $\nu\, (\mathrm{Hz})$, and $QE_\mathrm{int}\, (-)$ is the interacting quantum efficiency representing the probability an incident photon is detected \cite{EMVA_1288}.

The actual number of observed interacting photons for any given integration time will vary randomly and is accurately modeled by the Poisson distribution \cite{janesick_2007}. Assuming the transfer of interacting photons to photoelectrons is one-to-one, variations in the number of photoelectrons generated in each pixel per integration time can therefore be modeled as a Poisson random variable $\mathscr P\sim\mathcal P(\mu_{e\text{-}})$. As the irradiance--which we shall refer to as the illumination level--increases, $\mu_{e\text{-}}\to\infty$ leading to the asymptotic result
\begin{equation}
\frac{\mathscr P-\mu_{e\text{-}}}{\sqrt{\mu_{e\text{-}}}}\overset{d}{\to}\mathcal N(0,1);
\end{equation}
however, for even relatively small values of $\mu_{e\text{-}}$, say $\mu_{e\text{-}}>30$, the error in this approximation is generally acceptable for applied purposes and so we can assume $\mathscr P\sim\mathcal N(\mu_{e\text{-}},\mu_{e\text{-}})$.

Additionally, dark noise present in the sensor is comprised of dark current shot noise and read noise with the read noise further comprised of several other noise sources, e.g.~source follower noise, reset noise, etc., so by the central limit theorem we have the reasonable model $\mathscr D\sim\mathcal N(\mu_\mathscr D,\sigma_\mathscr D^2)$, where $\mu_\mathscr D$ and $\sigma_\mathscr D$ represent the sensor bias and dark noise in units of electrons, respectively. We will further assume $\mathscr P$ and $\mathscr D$ are independent. In considering the combined photon induced and dark signal $\mathscr P+\mathscr D$, one might be tempted to assume so long as $\mu_{e\text{-}}$ is large, that the approximate normality of $\mathscr P$ implies $\mathscr P+\mathscr D$ must also be approximately normal; however, this is not always the case. Once again using $\mathscr P\sim\mathcal P(\mu_{e\text{-}})$ the density of $\mathscr P+\mathscr D$ can be formally expressed by the convolution integral $f_{\mathscr P+\mathscr D}(x)=\int_{\mathbb N_0} f_\mathscr D(x-n)\,\mathrm dF_\mathscr P(n)$ leading to the explicit form (see \cite{starkey_2016} and Section 7.2 of \cite{janesick_2007}):
\begin{equation}
\label{eq:fPD_density}
f_{\mathscr P+\mathscr D}(x)=
\sum_{n=0}^\infty\phi(x-n;\mu_\mathscr D,\sigma_\mathscr D)\frac{e^{-\mu_{e\text{-}}}\mu_{e\text{-}}^n}{n!},
\end{equation}
where $\phi(\cdot;\mu,\sigma)$ is the normal probability density with mean $\mu$ and standard deviation $\sigma$. Due to the emerging importance of (\ref{eq:fPD_density}) in the literature \cite{starkey_2016,Nakamoto_2022} we also note that an alternative expression can be obtained by writing $f_{\mathscr P+\mathscr D}$ in terms of the exponential square series function \cite{Schmidt_2017}
\begin{equation}
    f_{\mathscr P+\mathscr D}(x)=\phi(x;\mu_\mathscr D,\sigma_\mathscr D)e^{-\mu_{e\text{-}}}E_{sq}(e^{-1/2\sigma_\mathscr D^2},\mu_{e\text{-}} e^{(x-\mu_\mathscr D)/\sigma^2},1),
\end{equation}
where $E_{sq}(q,r,z)\coloneqq\sum_{n=0}^\infty q^{n^2}r^n z^n/n!$. Proposition $5.2$ in Schmidt (2017) along with the identities $e^{iz}=\cos z+i\sin z$ and $2\cos z=e^{iz}+e^{-iz}$ give
\begin{multline}
    \label{eq:Esq_integral_rep}
    E_{sq}(e^{-1/2\sigma_\mathscr D^2},\Omega,1)=\\
    2\int_0^\infty\phi(t;0,1)\exp(\Omega\cos (t/\sigma_\mathscr D))\cos(\Omega\sin (t/\sigma_\mathscr D))\,\mathrm dt,
\end{multline}
which subsequently provides a novel integral representation for (\ref{eq:fPD_density}).

Figure \ref{fig:poisson_gaussian_mixture} plots $f_{\mathscr P+\mathscr D}$ for $\mu_{e\text{-}}=30\,e\text{-}$, $\mu_\mathscr D=0\,e\text{-}$ and $\sigma_\mathscr D=0.3,1.0\,e\text{-}$. Upon inspection, we see that despite $\mu_{e\text{-}}$ being large, so that $\mathscr P$ is approximately normal, only the density for $\mathscr P+\mathscr D$ corresponding to $\sigma_\mathscr D=1.0\,e\text{-}$ can be accurately modeled as normal. As such, on top of the restriction $\mu_{e\text{-}}>30\,e\text{-}$ we will also assume $\sigma_\mathscr D>1.0\,e\text{-}$ so that we may use the model $\mathscr P+\mathscr D\sim\mathcal N(\mu_{e\text{-}}+\mu_\mathscr D,\mu_{e\text{-}}+\sigma_\mathscr D^2)$. This additional assumption excludes photon counting devices such as Deep-Sub-Electron-Read Noise ({\sc dsern}) image sensors. We do note that in the case where a large number of sample are collected these assumptions can be loosened as the distributions of the sample statistics, e.g.~mean and variance, will agree with our model even if $\mathscr D$ and $\mathscr P+\mathscr D$ deviate from normality.
\begin{figure}[htb]
\centering
\includegraphics[scale=1]{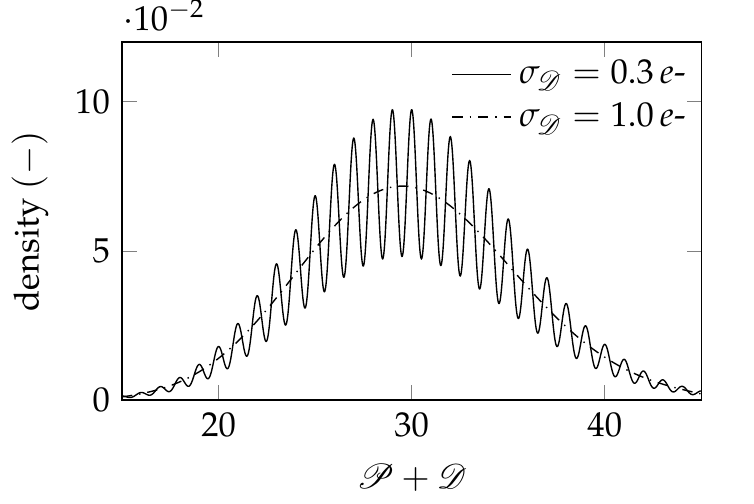}
\caption{Probability density of $\mathscr P+\mathscr D$ for $\mu_{e\text{-}}=30\,e\text{-}$, $\mu_\mathscr D=0.0\,e\text{-}$, and $\sigma_\mathscr D=0.3\,e\text{-}$ (solid) and $\sigma_\mathscr D=1.0\,e\text{-}$ (dash-dot).}
\label{fig:poisson_gaussian_mixture}
\end{figure}

In addition to being accurate for data from real image sensors, the normal model
\begin{gather}
\mathscr D\sim\mathcal N(\mu_\mathscr D,\sigma_\mathscr D^2)\\ 
\mathscr P+\mathscr D\sim\mathcal N(\mu_{e\text{-}}+\mu_\mathscr D,\mu_{e\text{-}}+\sigma_\mathscr D^2)
\end{gather}
is mathematically convenient due to the tractability of normal moments. However, the continuous signals $\mathscr D$ and $\mathscr P+\mathscr D$ are never directly observed because they are quantized via an analog-to-digital converter. This quantization step transforms the corresponding continuous density functions into discrete probability mass functions, which distorts the shape of the distributions and thus alters the moments. By further imposing $g\leq\sigma_\mathscr D$, the effects of quantization are negligible so that we can reasonably assume a normal model for the quantized signal as well \cite{janesick_2001,janesick_2007}.


\section{Review of gain estimation theory: shot-noise-limited case}
\label{sec:shot_noise_limited}

Our first goal is to present the major statistical results of conversion gain estimation for the special case of a shot-noise-limited response. The results in this section will serve as a starting point for the general case of gain estimation in Section \ref{sec:sub_shot_noise_limited}. While we focus on $g$-estimation for a single pixel by repeated sampling of the pixel in time, the analysis is equally valid in the case of spatially sampling an array of identical pixels exposed to a uniform light source. To help keep things organized, Table \ref{tab:SNL_symbols} lists key symbols pertaining to shot-noise-limited estimation and their associated formulae. Note that all these symbols are built up from only three fundamental quantities: $\mu_{e\text{-}}$, $g$, and $n$.

\begin{table}[htb]
\centering
\caption{\bf List of symbols and corresponding formulae associated with shot-noise-limited estimation.}
\begin{tabular}{cc|cc}
\hline
symbol &formula &symbol &formula \\
\hline
$\mu_P$ &$\mu_{e\text{-}}/g$ &$\alpha$ &$(n-1)/2$ \\
$\sigma_P^2$ &$\mu_{e\text{-}}/g^2$ &$\beta$ &$\alpha/\sigma_P^2$ \\
\hline
$\mu_{\bar P}$ &$\mu_P$ &$g$ &$\mu_P/\sigma_P^2$ \\
$\sigma_{\bar P}^2$ &$\sigma_P^2/n$ &$-$ &$-$ \\
\hline
\end{tabular}
 \label{tab:SNL_symbols}
\end{table}


\subsection{Estimator derivation}

A pixel can be modeled as a transfer function $\mathcal T:e\text{-}\to\mathrm{DN}$ mapping photoelectrons to a digital number output. In general, each pixel comprising the active sensor array is assumed to have it's own unique transfer function. Suppose $\mathcal T(e\text{-})=e\text{-}/g$ and our pixel has zero bias and dark noise so that the dark signal can be formally represented by the degenerate variable $\mathscr D\sim\delta(0)$. Because the only noise in the pixel output will come from photon shot noise we say the pixel exhibits a shot-noise-limited response. If $P=\mathcal T(\mathscr P)$ is the random variable representing the photon induced output signal in $\mathrm{DN}$, then it's easy to see
\begin{equation}
\label{eq:photon_induced_mean}
\mu_P\coloneqq\mathsf EP=\mathsf E(\mathscr P/g)=\mu_{e\text{-}}/g
\end{equation}
and
\begin{equation}
\label{eq:photon_induced_variance}
\sigma_P^2\coloneqq\mathsf{Var}P=\mathsf{Var}(\mathscr P/g)=\mu_{e\text{-}}/g^2.
\end{equation}
Combining these two results give us the fundamental photon transfer relation
\begin{equation}
g=\mu_P/\sigma_P^2.
\end{equation}

This fundamental relation implies a natural estimator for $g$. Let $\{P_1,\dots,P_n\}$ be a sample of $n$ i.i.d.~observations of our pixel exposed to some constant level of incident illumination for a fixed, nonzero integration time. Then we can estimate $g$ with
\begin{equation}
\label{eq:shot_noise_limited_estimator}
G=\bar P/\hat P,
\end{equation}
where $\bar P=\frac{1}{n}\sum_{k=1}^nP_k$ and $\hat P=\frac{1}{n-1}\sum_{k=1}^n(P_k-\bar P)^2$ are the sample mean and sample variance, respectively. Under the normal model $P_k\sim\mathcal N(\mu_P,\sigma_P^2)$, $(\bar P,\hat P)$ is a complete sufficient statistic of the unknown parameter $(\mu_P,\sigma_P^2)$ so that $G$ also happens to be the Uniformly Minimum-Variance Unbiased Estimator ({\sc umvue}) of its expected value \cite{casella_2002}.


\subsection{Historical developments}

Statistical analysis of the estimator (\ref{eq:shot_noise_limited_estimator}) has been previously conducted by Beecken \& Fossum (1996) as well as Janesick (2001) \cite{Beecken:96,janesick_2001}. In both works the moments of $G$ were approximated with the moments of it's first-order Taylor polynomial about $(\mathsf E\bar P,\mathsf E\hat P)=(\mu_P,\sigma_P^2)$
\begin{equation}
\label{eq:G_SNL_Taylor_series}
    G\approx g+\frac{g}{\mu_P}(\bar P-\mu_P)-\frac{g}{\sigma_P^2}(\hat P-\sigma_P^2).
\end{equation}
Using these approximate moments Beecken \& Fossum were able to show under the normal model of sensor noise (c.f.~Eq.~20 in \cite{Beecken:96} using $S_g/g\mapsto\mathsf{ACV}G$, $g\mapsto1/g$, $\bar x\mapsto\mu_P$, $N\mapsto n$, and $\sigma/S\mapsto 1$)
\begin{equation}
\label{eq:VarG_shot_limit_approx}
\mathsf{ACV}^2G\approx\frac{2}{n-1}+\frac{1}{n}\frac{1}{\mu_{e\text{-}}}.
\end{equation}
For clarity we note that the paper by Beecken \& Fossum actually studied the estimator $G^{-1}=\hat P/\bar P$, which is the conversion gain in units of $\mathrm{DN}/e\text{-}$; however, applying their noise model and statistical analysis to $G$ as given by (\ref{eq:shot_noise_limited_estimator}) gives the result in (\ref{eq:VarG_shot_limit_approx}). Furthermore, as $n$ becomes large we may replace $2/(n-1)$ with $2/n$, which is the same estimate given by Janesick (c.f.~Eq.~2.18 in \cite{janesick_2001} using $\sigma_K^2\mapsto\mathsf{Var}G$, $N_{pix}\mapsto n$, $S(\mathrm{DN})\mapsto\mu_P$, and $K\mapsto g$).

In both works it was noted that at typical illumination levels where $g$ is measured, (\ref{eq:VarG_shot_limit_approx}) is very well approximated by its first term, which happens to be the first-order Taylor approximation of $\mathsf{ACV}^2\hat P^{-1}$. In other words, for sufficiently large illumination, we have the approximate relation $\mathsf{ACV}G\approx\mathsf{ACV}\hat P^{-1}$. Such an approximation is useful because it tells us that the number of samples needed to measure $g$ to a given uncertainty can be approximated by the number of samples needed to measure $1/\sigma_P^2$ to the same uncertainty. This is a key insight that will appear several more times throughout this correspondence. Using the high illumination approximation $\mathsf{ACV}G\approx \sqrt{2/n}$ \cite[c.f.~Eq.~$6.12$,]{janesick_2007} we subsequently obtain Janesick's approximation for the optimal (minimal) number of samples needed to estimate $g$ to a desired relative uncertainty $\mathsf{acv}_0$:
\begin{equation}
\label{eq:SNL_OptSample_Janesick}
    n^\mathrm{opt}\approx \frac{2}{\mathsf{acv}_0^2}.
\end{equation}


\subsection{Further developments}

Statistical analysis of the shot-noise-limited estimator in (\ref{eq:shot_noise_limited_estimator}) yields tractable results for the density function, moments, and optimal sample size without the need to invoke approximate methods. For a normal model, we have for the distributions of the photon induced signal and it's sample statistics : $P_k\sim\mathcal N(\mu_P,\sigma_P^2)$, $\bar P\sim\mathcal N(\mu_{\bar P},\sigma_{\bar P}^2)$, and $\hat P\sim\mathcal G(\alpha,\beta)$, with the latter being a gamma variable parameterized in terms of shape $\alpha$ and rate $\beta$ (see Table \ref{tab:SNL_symbols} for parameter formulae). Since the $P_k$ are normal, $\bar P$ and $\hat P$ are also independent. For the density function we use change of variables to write $f_G(g)=\int_0^\infty t\phi(gt;\mu_{\bar P},\sigma_{\bar P})f_{\hat P}(t)\,\mathrm dt$, which after substituting $u=|g|t/\sigma$ gives
\begin{equation}
    f_G(g)=\frac{\alpha}{|g|}\left(\frac{\beta\sigma_{\bar P}}{|g|}\right)^\alpha\frac{e^{z^2(g)/4-\mu_{\bar P}^2/(2\sigma_{\bar P}^2)}}{\sqrt{2\pi}}D_{-\alpha-1}(z(g)).
\end{equation}
Here, $z(g)=\beta\sigma/|g|-\mu_{\bar P}\operatorname{sign}(g)/\sigma_{\bar P}$ and $D_\nu(z)\coloneqq\frac{e^{-z^2/4}}{\Gamma(-\nu)}\int_0^\infty t^{-\nu-1}e^{-t^2/2-zt}\,\mathrm dt$, which is the parabolic cylinder function.

As for the moments of $G$ we have by the independence of $\bar P$ and $\hat P$
\begin{equation}
    \mathsf EG^k=(\mathsf E\bar P^k)\mathsf E\hat P^{-k}.
\end{equation}
The moments of $\bar P$ are easily found by comparing the generating function $e^{2zt-t^2}=\sum_{n=0}^\infty H_n(z)t^n/n!$ to the moment generating function of $\bar P$ yielding
\begin{equation}
\mathsf E\bar P^k=(i\sigma_{\bar P}/\sqrt 2)^k H_k\left(-i\frac{\mu_{\bar P}}{\sqrt 2\sigma_{\bar P}}\right),
\end{equation}
where $H_k(z)\coloneqq (2z-\partial_z)^k\cdot 1$ denotes the $k$th degree Hermite polynomial and $i$ the imaginary unit. Likewise, we have for the moments of $\hat P^{-1}$
\begin{equation}
    \mathsf E\hat P^{-k}=\beta^k(\alpha)_{-k}
\end{equation}
with $(s)_n\coloneqq\Gamma(s+n)/\Gamma(s)$ denoting the Pochhammer symbol.

From here, the approximate results of Beecken, Fossum, and Janesick can be derived rigorously from an exact expression for $\mathsf{ACV}G$. The following lemma will aid us in this goal and also be used extensively throughout the rest of this work. All proofs can be found in Section \ref{sec:proofs}.
\begin{lemma}
\label{lem:ACV2_of_product}
Let $T=XY$. If $X$ and $Y$ are independent then
\begin{equation}
\mathsf{ACV}^2T=\mathsf{ACV}^2X+(\mathsf{ACV}^2X)(\mathsf{ACV}^2Y)+\mathsf{ACV}^2Y,
\end{equation}
with $\mathsf{ACV}T\coloneqq\sqrt{\mathsf{Var}T}/|\mathsf ET|$.
\end{lemma}
With the help of Lemma \ref{lem:ACV2_of_product} and the moment expressions given above we deduce the exact expression (c.f.~(\ref{eq:VarG_shot_limit_approx}))
\begin{equation}
\mathsf{ACV}^2G=\frac{2}{n-5}+\frac{2}{n(n-5)}\frac{1}{\mu_{e\text{-}}}+\frac{1}{n}\frac{1}{\mu_{e\text{-}}}.
\end{equation}
Setting $\mathsf{ACV}^2G=\mathsf{acv}_0^2$ yields a quadratic equation in $n$, which upon solving gives the optimal sample size needed to measure $g$ to a desired relative uncertainty $\mathsf{acv}_0$:
\begin{equation}
\label{eq:shot_noise_lim_G_est_opt_sample}
n^\mathrm{opt}=\frac{2+5\mathsf{acv}_0^2+\frac{1}{\mu_{e\text{-}}}+\left((2+5\mathsf{acv}_0^2+\frac{1}{\mu_{e\text{-}}})^2-12\frac{\mathsf{acv}_0^2}{\mu_{e\text{-}}}\right)^{1/2}}{2\mathsf{acv}_0^2}.
\end{equation}
Unfortunately, this exact expression is not of great use in practice because it depends on the unknown quantity $\mu_{e\text{-}}$, which also cannot be directly measured without \emph{a priori} knowledge of $g$. To obtain an approximation that is independent of $\mu_{e\text{-}}$, we first consider the following result showing that $\mathsf{ACV}G$ is dominated by $\mathsf{ACV}\hat P^{-1}$ at high-illumination.
\begin{theorem}
\label{thm:ACVG_shot_limited_approx}
Let $G$ be as given in (\ref{eq:shot_noise_limited_estimator}). As illumination increases, $\mu_{e\text{-}}\to\infty$ and
\begin{equation}
\mathsf{ACV}G=\mathsf{ACV}\hat P^{-1}\left(1+\frac{n-3}{4n}\frac{1}{\mu_{e\text{-}}}+\mathcal O(\mu_{e\text{-}}^{-2})\right),
\end{equation}
with $\mathsf{ACV}\hat P^{-1}=\sqrt{2/(n-5)}$.
\end{theorem}
Theorem \ref{thm:ACVG_shot_limited_approx} confirms the observations of \cite{Beecken:96,janesick_2001} in that
\begin{equation}
\mathsf{ACV}G\sim\mathsf{ACV}\hat P^{-1}
\end{equation}
at high-illumination. Because of this finding, the optimal sample size for $G$ can be approximated by the optimal sample size for $\hat P^{-1}$ when restricted to high-illumination conditions. Setting $\mathsf{ACV}^2\hat P^{-1}=\mathsf{acv}_0^2$ yields a linear equation in $n$, which upon solving for $n$ subsequently gives us the high-illumination, asymptotic approximation for the optimal sample size of $G$
\begin{equation}
\label{eq:shot_noise_limited_optimal_samples}
n^\mathrm{opt}\sim\frac{2}{\mathsf{acv}_0^2}+5,\quad \mu_{e\text{-}}\to\infty.
\end{equation}
For example, choosing a desired relative uncertainty of $1\%$ we have $\mathsf{acv}_0=0.01$ and $n^\mathrm{opt}\approx 20\, 005$, which agrees with Janesick's approximation of $n^\mathrm{opt}\approx 20\, 000$ as given by (\ref{eq:SNL_OptSample_Janesick}). Unfortunately, all of the results in this section break down when sensor dark noise is non-negligible and thus have limited applications. We are now ready to move onto the more general case.


\section{Review of gain estimation theory: general case}
\label{sec:sub_shot_noise_limited}

As was the case in the previous section, general gain estimation requires the use of many symbols that can be combined and manipulated. To stay organized, Table \ref{tab:GEN_symbols} lists many of the key symbols used along with their corresponding formulae. Note that all of these symbols are constructed from six fundamental quantities: $\mu_{e\text{-}}$, $\mu_\mathscr D$, $\sigma_\mathscr D^2$, $g$, $n_1$, and $n_2$.

\begin{table}[htb]
\centering
\caption{\bf List of symbols and corresponding formulae associated with general estimation.}
\begin{tabular}{cc|cc}
\hline
symbol &formula &symbol &formula \\
\hline
$\mu_D$ &$\mu_\mathscr D/g$ &$\mu_P^2$ &$\mu_{e\text{-}}/g$ \\
$\sigma_D^2$ &$\sigma_\mathscr D^2/g^2$ &$\sigma_P^2$ &$\mu_{e\text{-}}/g^2$ \\
\hline
$\mu_{P+D}$ &$\mu_P+\mu_D$ &$\alpha_1$ &$(n_1-1)/2$ \\
$\sigma_{P+D}^2$ &$\sigma_P^2+\sigma_D^2$ &$\alpha_2$ &$(n_2-1)/2$ \\
\hline
$\mu_{\bar P}$ &$\mu_P$ &$\beta_1$ &$\alpha_1/\sigma_{P+D}^2$ \\
$\sigma_{\bar P}^2$ &$\sigma_{P+D}^2/n_1+\sigma_D^2/n_2$ &$\beta_2$ &$\alpha_2/\sigma_D^2$ \\
\hline
$\mu_{\hat P}$ &$\sigma_P^2$ &$g$ &$\frac{\mu_{P+D}-\mu_D}{\sigma_{P+D}^2-\sigma_D^2}$ \\
$\sigma_{\hat P}^2$ &$\alpha_1/\beta_1^2+\alpha_2/\beta_2^2$ &$-$ &$-$ \\
\hline
\end{tabular}
 \label{tab:GEN_symbols}
\end{table}


\subsection{Estimator derivation}

We now consider the more general case where the pixel exhibits both a bias and non-negligible dark noise: $\mathscr D\sim\mathcal N(\mu_\mathscr D,\sigma_\mathscr D^2)$ with $\mu_\mathscr D,\sigma_\mathscr D\neq 0$. We again note that in general, each pixel in the active sensor array will exhibit unique values of $\mu_\mathscr D$ and $\sigma_\mathscr D$. By linearity of the transfer function, the digital output of a pixel in the absence of illumination for some fixed, nonzero integration time is $D=\mathcal T(\mathscr D)$ with
\begin{equation}
    \mu_D\coloneqq\mathsf ED=\mu_\mathscr D/g
\end{equation}
and
\begin{equation}
   \sigma_D^2\coloneqq\mathsf{Var}D=\sigma_\mathscr D^2/g^2.
\end{equation}
Likewise, the digital output of the same pixel exposed to incident illumination for the same fixed, nonzero integration time gives the combined dark and photon induced signal $P+D=\mathcal T(\mathscr P+\mathscr D)=\mathcal T(\mathscr P)+\mathcal T(\mathscr D)$, which by the assumed independence of $\mathscr P$ and $\mathscr D$ further gives
\begin{equation}
    \mu_{P+D}\coloneqq\mathsf E(P+D)=\mu_P+\mu_D
\end{equation}
and
\begin{equation}
    \label{eq:dark_plus_shot_noise_relation}
    \sigma_{P+D}^2\coloneqq\mathsf{Var}(P+D)=\sigma_P^2+\sigma_D^2.
\end{equation}
Noting that $\mu_P=\mu_{P+D}-\mu_D$ and $\sigma_P^2=\sigma_{P+D}^2-\sigma_D^2$ leads to the modified gain relation
\begin{equation}
\label{eq:modified_gain_relation}
g=(\mu_{P+D}-\mu_D)/(\sigma_{P+D}^2-\sigma_D^2).
\end{equation}
As before, this relation suggests a natural estimator for $g$. Let $\{X_1,\dots,X_{n_1}\}$ be a sample of $n_1$ i.i.d.~observations of a pixel exposed to a constant level of illumination for a fixed, nonzero integration time and $\{Y_1,\dots,Y_{n_2}\}$ be an independent sample of $n_2$ i.i.d.~observations of the same pixel in the absence of illumination for the same fixed, nonzero integration time. Then we can estimate $g$ for our pixel with
\begin{equation}
\label{eq:sub_shot_noise_limited_estimator}
G=(\bar X-\bar Y)/(\hat X-\hat Y),
\end{equation} 
with bar $(\bar\cdot)$ and hat $(\hat\cdot)$ accents again denoting sample means and sample variances, respectively. In this context we see
\begin{equation}
\bar P\coloneqq\bar X-\bar Y    
\end{equation}
estimates the photon induced signal $\mu_P$, while
\begin{equation}
\hat P\coloneqq\hat X-\hat Y
\end{equation}
estimates the photon induced signal variance $\sigma_P^2$, in a similar manner to that of the estimators $\bar P$ and $\hat P$ in (\ref{eq:shot_noise_limited_estimator}). Note that now $\bar P$ no longer represents a single sample mean but rather the difference of sample means and likewise for $\hat P$.


\subsection{Historical developments}

Janesick (2007) was the first work to include the contributions of dark noise in conversion gain estimator uncertainty. He did this by studying a variant of (\ref{eq:sub_shot_noise_limited_estimator}) where the numerator is substituted with the random variable $\bar P$ ($\mathsf E\bar P=\mu_P$, $\mathsf{Var}\bar P=\sigma_P^2/n_1)$. The relative uncertainty of the resulting estimator was then approximated via the moments of its first-order Taylor polynomial \cite{janesick_2007}. Further imposing large and equal numbers of samples for both the bright and dark collection, $n_1=n_2=n$, Janesick obtained the approximation (c.f.~Eq.~6.7-6.8 in \cite{janesick_2007} using $N_\mathrm{PIX}\mapsto n$, $K\mapsto g$, $S\mapsto \mu_P$, $N\mapsto\sigma_{P+D}$, and $R\mapsto\sigma_D$)
\begin{equation}
\label{eq:relative uncertainty_general_estimator_Janesick_approx}
\mathsf{ACV}^2G\approx\frac{2}{n}+\frac{1}{n}\frac{1}{\mu_{e\text{-}}}+\frac{4}{n}\frac{\sigma_\mathscr D^2}{\mu_{e\text{-}}}+\frac{4}{n}\left(\frac{\sigma_\mathscr D^2}{\mu_{e\text{-}}}\right)^2.
\end{equation}
It was noted that when photon shot noise is large ($\mu_{e\text{-}}\to\infty)$ and dominates sensor dark noise ($\mu_{e\text{-}}\gg\sigma_\mathscr D^2$), this estimate reduces to the high illumination, shot-noise-limited estimate $\mathsf{ACV}G\approx\sqrt{2/n}$ discussed in Section \ref{sec:shot_noise_limited}.

Following Janesick's work, Hendrickson (2017) was the first attempt to draw exact statistical conclusions about the full estimator (\ref{eq:sub_shot_noise_limited_estimator}) by applying the normal model $X_k\sim\mathcal N(\mu_{P+D},\sigma_{P+D}^2)$ and $Y_k\sim\mathcal N(\mu_D,\sigma_D^2)$ to derive the density of $G$ in the form of the Centralized Inverse-Fano ({\sc cif}) distribution
\begin{equation}
f_G(g)=\int_{-\infty}^\infty |t|\phi(gt;\mu_{\bar P},\sigma_{\bar P})f_{\hat P}(t)\,\mathrm dt,
\end{equation}
where
\begin{equation}
    f_{\hat P}(t)=C\times
    \begin{cases}
    \frac{e^{\beta_2t}}{\Gamma(\alpha_2)}U\left({1-\alpha_2\atop 2-\alpha_1-\alpha_2};-(\beta_1+\beta_2)t\right) &t<0\\
    \frac{e^{-\beta_1t}}{\Gamma(\alpha_1)}U\left({1-\alpha_1\atop 2-\alpha_1-\alpha_2};(\beta_1+\beta_2)t\right) &t\geq 0,
    \end{cases}
\end{equation}
with $C=\beta_1^{\alpha_1}\beta_2^{\alpha_2}(\beta_1+\beta_2)^{1-\alpha_1-\alpha_2}$ is the gamma-difference distribution \cite{mathai_1993,krishna_2011,klar_2015,hancova_2022}, and $U(a,b,z)$ is Kummer's confluent hypergeometric function of the 2nd-kind \cite{Hendrickson:17}. Under this model it was shown that $G$ has ill-defined moments which is due to the fact that the tails of the probability density $f_G$ decay like those of the Cauchy density \cite{Lehmann_1988}. For this reason, Hendrickson (2019) \cite{hendrickson_2019} extended the notion of statistical moments in the same manner as Peng \cite{PENG_2008,PENG_2013} by deriving the first moment of $G$ in the sense of the Cauchy principal value $\mathsf E_\mathcal PG=\lim_{R\to\infty}\int_{-R}^Rtf_G(t)\,\mathrm dt$
\begin{multline}
\label{eq:first_pseudo_moment_of_G}
\mathsf E_\mathcal PG=g\frac{(\frac{\alpha_1}{\beta_1}-\frac{\alpha_2}{\beta_2})\beta_1^{\alpha_1}\beta_2^{\alpha_2}(\beta_1+\beta_2)^{1-\alpha_1-\alpha_2}}{(\alpha_1+\alpha_2-1)\operatorname{B}(\alpha_1,\alpha_2)}\biggl(\\%
\psi(\alpha_1)-\log\beta_1+\frac{(\alpha_1-1)\beta_2}{\alpha_2\beta_1}{_3F_2}\left({2-\alpha_1,1,1\atop 1+\alpha_2,2};-\frac{\beta_2}{\beta_1}\right)\\%
-\psi(\alpha_2)+\log\beta_2-\frac{(\alpha_2-1)\beta_1}{\alpha_1\beta_2}{_3F_2}\left({2-\alpha_2,1,1\atop 1+\alpha_1,2};-\frac{\beta_1}{\beta_2}\right)%
\biggr).
\end{multline}
Here, $\log z$ is the natural logarithm, $\psi(z)$ is the digamma function, $\operatorname{B}(\alpha,\beta)$ is the beta function, and ${_pF_q}(\mathbf a;\mathbf b;z)$ is the generalized hypergeometric function. It was shown that $\mathsf E_\mathcal PG$ agrees with actual sample means of conversion gain data when $\mathsf P(\hat P\leq 0)\approx 0$. Additionally, Hendrickson (2021) showed that no unbiased, finite variance estimator of the modified gain relation (\ref{eq:modified_gain_relation}) exists for all possible parameters under the normal model of noise despite $(\bar Y,\hat Y,\bar X,\hat X)$ constituting a complete sufficient statistic for the parameter $(\mu_D,\sigma_D^2,\mu_{P+D},\sigma_{P+D}^2)$ \cite[Theorem 14,]{hendrickson_2021}.


\subsection{Further developments}

The nonexistence of $G$'s moments lies in the fact that $\hat P$ has positive and continuous probability density at zero, which manifests as non-integrable singularities in the integral representations of these moments. To assign some notion of higher-order moments to $G$ we may take advantage of the independence of $\bar P$ and $\hat P$ and the concept of regularization to define pseudomoments as
\begin{equation}
\mathsf E_\mathcal PG^k\coloneqq(\mathsf E\bar P^k)\mathsf E_\mathcal P\hat P^{-k},\quad k\in\mathbb N,
\end{equation}
where $\mathcal P$ denotes the principal-value regularization
\begin{equation}
\label{eq:pseudomoment_definition}
    \mathsf E_\mathcal P\hat P^{-k}\coloneqq\lim_{\epsilon\to 0^+}\int_{\Bbb R\setminus(-\epsilon,\epsilon)}\frac{f_{\hat P}(t)}{t^k}\,\mathrm dt-h_k(\epsilon),
\end{equation}
with $h_1(\epsilon)=0$ and
\begin{equation}
    h_k(\epsilon)=\sum_{\ell=0}^{k-2}\frac{f_{\hat P}^{(\ell)}(0)}{\ell!}\left(\frac{1-(-1)^{k-\ell-1}}{(k-\ell-1)\epsilon^{k-\ell-1}}\right)
\end{equation}
for $k\geq 2$ \cite{fox_1957,Galapon_2016}. Alternative definitions and methods also exist for evaluating the pseudomoments. For example, we may express them via a moment generating function as
\begin{equation}
\mathsf E_\mathcal P\hat P^{-k}= \frac{1}{(k-1)!}\partial_\omega^{k-1}\mathcal H[f_{\hat P}](\omega)\Big|_{\omega=0}
\end{equation}
with
\begin{equation}
\label{eq:density_Hilbert_xform}
\mathcal H[f_{\hat P}](\omega)=\lim_{\epsilon\to 0^+}\int_{\mathbb R\setminus(\omega-\epsilon,\omega+\epsilon)}\frac{f_{\hat P}(t)}{t-\omega}\,\mathrm dt,
\end{equation}
denoting the Hilbert transform of the density $f_{\hat P}$ \cite{criscuolo_1997}. Regardless of the method used to evaluate them, the subscript $\mathcal P$ is there to remind the reader that these moments are regularized and do not exist in the traditional sense because the integrals representing them diverge.

Evaluating the pseudomoments proves to be quite challenging.  As pointed out earlier, the work in \cite{hendrickson_2019} was able to obtain an expression for the special case $k=1$, as given in (\ref{eq:first_pseudo_moment_of_G}), through the use of complex methods involving contour integration. For $k\geq 2$, one could implement (\ref{eq:pseudomoment_definition}) numerically; however, when $n_1$ and $n_2$ become even moderately large the evaluation of $f_{\hat P}$ becomes unstable causing numerical computation to fail.  In such a case we may approximate the distribution of $\hat P$ by a normal one with mean $\mu_{\hat P}$ and variance $\sigma_{\hat P}^2$ so that we may use the normal approximation given by Quenouille \cite{QUENOUILLE_1956}
\begin{equation}
\label{eq:negative_mgf_normal}
\mathcal H[f_{\hat P}](\omega)\sim\frac{\sqrt 2}{\sigma_{\hat P}}\mathcal D\left(\frac{\mu_{\hat P}-\omega}{\sqrt 2 \sigma_{\hat P}}\right).
\end{equation}
Here, $\mathcal D(z)\coloneqq e^{-z^2}\int_0^z e^{t^2}\,\mathrm dt$ denotes the Dawson integral. Higher-order derivatives for the Dawson integral are given by Barakat (1971), which allow us to deduce a closed-form, asymptotic approximation for the pseudomoments of $\hat P^{-1}$ given large $n_1$ and $n_2$ \cite{BARAKAT_1971}:
\begin{equation}
    \mathsf E_\mathcal P\hat P^{-k}\sim\frac{2}{(k-1)!}\frac{1}{(\sqrt 2\sigma_{\hat P})^k}(H_{k-1}(z_{\hat P})\mathcal D(z_{\hat P})-P_{k-2}(z_{\hat P})),
\end{equation}
where $z_{\hat P}=\mu_{\hat P}/(\sqrt 2\sigma_{\hat P})$, and $P_n$ is a polynomial satisfying $P_n(t)=2tP_{n-1}(t)-2nP_{n-2}(t)$ with $P_{-1}(t)=0$ and $P_0(t)=1$. Combining this result with the moments of $\bar P$ subsequently give us the large sample size asymptotic approximation of the pseudomoments for $G$
\begin{multline}
\label{eq:higher_order_pseudo_moment_of_G_aprx}
\mathsf E_\mathcal PG^k\sim\frac{2}{(k-1)!}\left(\frac{i\sigma_{\bar P}}{2\sigma_{\hat P}}\right)^k H_k(-iz_{\bar P})\\
\cdots\times(H_{k-1}(z_{\hat P})\mathcal D(z_{\hat P})-P_{k-2}(z_{\hat P})),
\end{multline}
where $z_{\bar P}=\mu_{\bar P}/(\sqrt 2\sigma_{\bar P})$. For example, in the case $k=1$ we obtain the large sample size (large $\alpha$) asymptotic approximation of (\ref{eq:first_pseudo_moment_of_G})
\begin{equation}
\label{eq:first_pseudo_moment_of_G_aprx}
\mathsf E_\mathcal PG\sim g\frac{\sqrt 2(\frac{\alpha_1}{\beta_1}-\frac{\alpha_2}{\beta_2})}{\sqrt{\frac{\alpha_1}{\beta_1^2}+\frac{\alpha_2}{\beta_2^2}}}\mathcal D\left(\frac{\frac{\alpha_1}{\beta_1}-\frac{\alpha_2}{\beta_2}}{\sqrt 2\sqrt{\frac{\alpha_1}{\beta_1^2}+\frac{\alpha_2}{\beta_2^2}}}\right).
\end{equation}

As a verification of the accuracy of these approximations, using the parameters $\mu_P=9\,\mathrm{DN}$, $\sigma_{P+D}^2=10\,\mathrm{DN}^2$, $\sigma_D^2=1\,\mathrm{DN}^2$ (so that $g=1$), $n_1=101$, and $n_2=51$ we calculated $\mathsf E_\mathcal PG$ using the exact expression (\ref{eq:first_pseudo_moment_of_G}) as well as the normal approximation (\ref{eq:first_pseudo_moment_of_G_aprx}) yielding $\mathsf E_\mathcal PG=1.02604$ and $\mathsf E_\mathcal PG\approx 1.02738$, respectively. This resulted in only a $0.13\%$ approximation error showing that the normal approximation in (\ref{eq:higher_order_pseudo_moment_of_G_aprx}) will approximate $\mathsf E_\mathcal PG$ as well as the higher-order pseudomoments for the chosen sample sizes. We typically deal with much larger sample sizes in {\sc pt} conversion gain estimation and therefore expect (\ref{eq:first_pseudo_moment_of_G_aprx}) to be a good approximation to the exact pseudomoments of $G$ in most scenarios.

Through the use of pseudomoments we can subsequently derive other quantities of interest for $G$ like the (pseudo) absolute coefficient of variation
\begin{equation}
\mathsf{ACV}_\mathcal P^2G=\mathsf{ACV}_\mathcal P^2\hat P^{-1}+(\mathsf{ACV}_\mathcal P^2\hat P^{-1})(\mathsf{ACV}^2\bar P)+\mathsf{ACV}^2\bar P
\end{equation}
where
\begin{equation}
\label{eq:pseudo_acv2}
\mathsf{ACV}^2_\mathcal P\hat P^{-1}=\frac{\mathsf E_\mathcal P\hat P^{-2}-(\mathsf E_\mathcal P\hat P^{-1})^2}{(\mathsf E_\mathcal P\hat P^{-1})^2},
\end{equation}
and $\mathsf{ACV}^2\bar P=\mathsf{Var}\bar P/(\mathsf E\bar P)^2$ is defined in the traditional sense. Likewise, we have the (pseudo) absolute relative bias
\begin{equation}
\mathsf{ARB}_\mathcal PG=\mathsf{ARB}_\mathcal P\hat P^{-1}
\end{equation}
with
\begin{equation}
\label{eq:pseudo_arb}
\mathsf{ARB}_\mathcal P\hat P^{-1}=\left|\frac{\mathsf E_\mathcal P\hat P^{-1}-(\mathsf E\hat P)^{-1}}{(\mathsf E\hat P)^{-1}}\right|.
\end{equation}

What remains is to address why these pseudomoments are useful for describing moments of actual data. The motivation for introducing pseudomoments was to assign analytical expressions to the moments $\mathsf E\hat P^{-k}$, which diverge on sets of the form $|\hat P|<\epsilon$.  This is clear from writing
\begin{equation}
\mathsf E\hat P^{-k}=\int_{\mathbb R\setminus(-\epsilon,\epsilon)}\frac{f_{\hat P}(t)}{t^k}\,\mathrm dt+\int_{-\epsilon}^\epsilon \frac{f_{\hat P}(t)}{t^k}\,\mathrm dt,
\end{equation}
where the first integral in this decomposition always converges, while the second integral always diverges for any choice of $k\in\mathbb N$ and $\epsilon>0$. As such, the principal value regularization provides a means of discarding the divergent terms arising out of the second integral to provide a finite expression for $\mathsf E\hat P^{-k}$. In practice we don't ever observe $\hat P$ out in the extreme tails of its assumed distribution due to the small probabilities of such events, and even more so, inaccuracies in our assumed noise model. In particular, recall that we assumed a normal distribution for the digital signal $X_k$ and $Y_k$ but in practice these quantities can only take on values between $0$ and $2^{N_{\mathrm{bits}}}-1$ with $N_{\mathrm{bits}}$ denoting the bit-depth of the analog-to-digital converter. Because the normal model assigns positive density outside this interval, we see the model inherently overestimates the tails of the actual data. Therefore, if $\mathsf{ACV}\hat P$ is small we won't observe $|\hat P|<\epsilon$ and the sample moments will agree with the pseudomoments. Again using the parameters following (\ref{eq:first_pseudo_moment_of_G_aprx}), $10^9$ pseudorandom observations of $\bar P$ and $\hat P$ were generated, which were then used to compute a sample of $10^9$ observations of $G$.  Due to a sufficiently small value of $\mathsf{ACV}\hat P$, all observed values of $\hat P$ were strictly positive, so we should expect the sample statistics to agree with the theoretical pseudomoments. Computing the sample mean yielded $\bar G=1.02604\dots$, which agreed with the exact value of $\mathsf E_\mathcal PG$ to six significant digits. Likewise, we would expect the higher-order sample moments to agree with their corresponding higher-order pseudomoments for these parameters.


\section{Optimal sample size pairs for conversion gain estimation}
\label{sec:sub_shot_noise_limited_optimal_sample sizes}

Now that we have a thorough understanding of the statistical characteristics for the general conversion gain estimator, we may begin to tackle the problem of optimal measurement. Recall from Section \ref{sec:shot_noise_limited} that the optimal sample size for the one-sample, shot-noise-limited estimator $G$ satisfied $\mathsf{ACV}G(n^\mathrm{opt})=\mathsf{acv}_0$ for any choice of $\mathsf{acv}_0>0$. In the present problem, we are now working with an estimator of two samples and must first define what is meant by optimal sample sizes in this two-sample case.
\begin{definition}[Optimal sample size pairs]
\label{def:optimal_sample_sizes}
Let $T=f(\mathbf X,\mathbf Y)$ be a statistic of a sample $\mathbf X$ of size $n_1$ and another sample $\mathbf Y$ of size $n_2$. Furthermore, let $\mathsf{ACV}T(n_1,n_2)$ denote the absolute coefficient of variation for $T$ as a function of the sample sizes. Then the optimal sample size pairs for $T$ shall be defined as the ordered pair $(n_1^\mathrm{opt},n_2^\mathrm{opt})$, which satisfies the system of equations
\begin{gather}
\inf_{n_2}\, \mathsf{ACV}T(N-n_2,n_2)\Big|_{N=n_1^{\mathrm{opt}}+n_2^{\mathrm{opt}},\,n_2=n_2^{\mathrm{opt}}}\label{eq:opt_soe_1}\\ 
\mathsf{ACV}T(n_1^{\mathrm{opt}},n_2^{\mathrm{opt}})=\mathsf{acv}_0.\label{eq:opt_soe_2}
\end{gather}
\end{definition}
In this system of equations we see (\ref{eq:opt_soe_1}) fixes the total number of samples to $N=n_1+n_2$ and solves for the $n_2$ that minimizes $\mathsf{ACV}T$. In the special case where $\mathsf{ACV}T(N-n_2,n_2)$ is strictly convex in $n_2$, this minimization can be solved via equating the derivative with zero: $\partial_{n_2}\mathsf{ACV}T(N-n_2,n_2)=0$. Solving (\ref{eq:opt_soe_1}) and then substituting $N\mapsto n_1^\mathrm{opt}+n_2^\mathrm{opt}$ and $n_2\mapsto n_2^\mathrm{opt}$ implicitly defines the optimal sample sizes as a function of each other, which we will call the \emph{optimality relation}. For example, in the proof of Lemma \ref{lem:barP_optimal_samples} we derive the optimality relation for the estimator $\bar P$ in the form $n_2^\mathrm{opt}=\frac{\sigma_D}{\sigma_{P+D}}n_1^\mathrm{opt}$, which expresses the relationship between the optimal sample sizes for $\bar P$. Substituting the optimality relation into (\ref{eq:opt_soe_2}) then scales the optimal sample sizes so that they not only satisfy the optimality relation but also achieve a prescribed final absolute coefficient of variation equal to $\mathsf{acv}_0$. In this way, $n_1^\mathrm{opt}$ and $n_2^\mathrm{opt}$ represent the sample sizes whose sum is the minimal possible number of total samples needed to force $\mathsf{ACV}T$ equal to $\mathsf{acv}_0$ and therefore serves as a good generalization of optimal sample size to the two-sample case.

In Section \ref{sec:sub_shot_noise_limited} we were unable to derive an exact expression for $\mathsf{ACV}_\mathcal PG$ and thus are forced to make some form of approximation to get a handle on this problem. Even if we did have an exact expression, substituting it in Definition \ref{def:optimal_sample_sizes} would almost certainly yield an intractable system of equations. To overcome these barriers, we will take the approach of linearizing $\hat P^{-1}$ by replacing it with it's first-order Taylor polynomial about $\mathsf E\hat P=\sigma_P^2$ and instead focus on
\begin{equation}
\label{eq:G_delta_estimator}
G_\delta=\bar P\times\hat P^{-1}_\delta,
\end{equation}
with
\begin{equation}
\label{eq:Pinv_delta_estimator}
\hat P^{-1}_\delta=\frac{1}{\sigma_P^2}-\frac{\hat P-\sigma_P^2}{\sigma_P^4}.
\end{equation}
The advantage we gain from this linearization is that $\hat P^{-1}_\delta$ and $G_\delta$ have simple and well-defined moments that can be used to make concrete conclusions about their statistical properties including their optimal sample sizes. Furthermore, so long as $\mathsf{ACV}\hat P$ is small we have $\hat P^{-1}_\delta\overset{d}{\approx}\hat P^{-1}$, which implies $G_\delta\overset{d}{\approx}G$ so that any conclusions we make about the random variables $\hat P^{-1}_\delta$ and $G_\delta$ apply to $\hat P^{-1}$ and $G$ when $\mathsf{ACV}\hat P$ is small. In particular, note that $\mathsf{ACV}\hat P=\mathsf{ACV}\hat P^{-1}_\delta$ so we can conclude that the optimal sample sizes for $\hat P^{-1}_\delta$ should be good approximations for those of $G$ when $\mathsf{acv}_0$ is chosen to be small and the optimal sample sizes of $\hat P^{-1}_\delta$ are approximately equal to those of $G_\delta$.

Our angle of attack from here will be as follows. We will first use Definition \ref{def:optimal_sample_sizes} to derive properties pertaining to the optimal sample sizes of $\bar P$ (Subsection \ref{subsec:analysis_of_barP}), which will be used later on to derive similar properties of the optimal sample sizes for $G_\delta$. Following analysis of $\bar P$ will be analogous study of $\hat P^{-1}_\delta$ (Subsection \ref{subsec:analysis_of_hatP_inv_delta}), where we derive the optimal sample sizes for $\hat P^{-1}_\delta$ and their properties. Subsection \ref{subsec:analysis_of_G_delta} then studies the estimator $G_\delta$ with the goal of understanding when the optimal sample sizes for $\hat P^{-1}_\delta$ are a good approximation for those of $G_\delta$. We first show that $\mathsf{ACV}G_\delta$ is dominated by $\mathsf{ACV}\hat P^{-1}_\delta$ at high-illumination (Theorem \ref{thm:ACVG_general_approx}) so that the optimal sample sizes for $\hat P^{-1}_\delta$ are asymptotically equal to those of $G_\delta$ at high-illumination. At the other end of the illumination range, Theorem \ref{thm:G_delta_asymptotic_optimal_samples} shows that the optimal sample sizes for $\hat P^{-1}_\delta$ are asymptotically proportional to those of $G_\delta$ with the constant of proportionality approaching one with increasing dark noise $\sigma_\mathscr D$. What this demonstrates is that the optimal sample sizes for $\hat P^{-1}_\delta$ should be excellent approximations to those of $G_\delta$ at any illumination level given sufficiently large dark noise. This observation leads us to construct a metric, $\mathcal E_\mathrm{opt}$, which indicates how good the optimal sample sizes of $\hat P^{-1}_\delta$ are as approximations for those of $G_\delta$ as a function of illumination level. By studying $\mathcal E_\mathrm{opt}$ at low-illumination we are able to put a rule-of-thumb on how large the dark noise must be to obtain a good approximation. In particular, results will show that for sensors with moderate dark noise of $\sigma_\mathscr D\geq 5\,e\text{-}$, the optimal sample sizes of $\hat P^{-1}_\delta$ are excellent approximations for those of $G_\delta$ at any illumination level. Theorem \ref{thm:G_opt_ARB_ACV} then confirms these findings by showing for small $\mathsf{acv}_0$
\begin{equation}
\mathsf{ACV}_\mathcal PG(n_1^\mathrm{opt},n_2^\mathrm{opt})=c\cdot\mathsf{acv}_0+\mathcal O(\mathsf{acv}_0^3),
\end{equation}
with $n_i^\mathrm{opt}$ again representing the optimal sample sizes for $\hat P^{-1}_\delta$ and $c\to 1$ with increasing illumination and/or increasing dark noise. Consequently, we show that the optimal sample sizes for $\hat P^{-1}_\delta$ serve as excellent approximations for those of $G$ at any illumination level when $\mathsf{acv}_0$ is small and $\sigma_\mathscr D\geq 5\,e\text{-}$.

Before proceeding, it will also be convenient to introduce the dimensionless quantity
\begin{equation}
\zeta=\sigma_D^2/\sigma_{P+D}^2.
\end{equation}
We will reparameterize all subsequent analysis in terms of this quantity via the following lemma.
\begin{lemma}
\label{lem:parameter_relations}
Under the assumed model we have $\mu_P=\sigma_D^2\zeta^{-1}(1-\zeta)g$ and $\sigma_P^2=\sigma_D^2\zeta^{-1}(1-\zeta)$.
\end{lemma}
Because $0<\sigma_D<\sigma_{P+D}$, it follows that $\zeta\in(0,1)$. As illumination decreases to zero, we find $\zeta\to 1^-$. Likewise, as illumination increases without bound $\zeta\to 0^+$, with $\zeta=0$ denoting a mathematical definition of the shot-noise-limit. Real sensors always contain some dark noise and are limited in well capacity; therefore, one can only achieve the shot-noise-limit in theory and never in real experiments. Lastly, the reader should take note that $\zeta$ is increasing as illumination decreases.


\subsection{Statistical analysis of bar-{\normalfont $P$}}
\label{subsec:analysis_of_barP}

We begin our analysis by studying the random variable $\bar P$ with the primary purpose of understanding its optimal sample sizes for later use in studying $G_\delta$. Using Lemma \ref{lem:parameter_relations}, the dark noise relation $\sigma_\mathscr D=\sigma_D\times g$, and the distributional result $\bar P\sim\mathcal N(\mu_P,\sigma_{P+D}^2/n_1+\sigma_D^2/n_2)$ we are able to write the squared absolute coefficient of variation as
\begin{equation}
\label{eq:ACV2_barP}
\mathsf{ACV}^2\bar P=\frac{1}{\sigma_\mathscr D^2}\frac{\zeta}{(1-\zeta)^2}\left(\frac{1}{n_1}+\frac{\zeta}{n_2}\right).
\end{equation}
With this expression in hand, we present some key properties of the optimal sample sizes for $\bar P$.
\begin{lemma}
\label{lem:barP_optimal_samples}
Let $n_1^\mathrm{opt}$ and $n_2^\mathrm{opt}$ denote the optimal sample sizes for $\bar P$. As the illumination decreases, $\zeta\to 1^-$, $n_2^\mathrm{opt}/n_1^\mathrm{opt}\to 1^-$, and $n_i^\mathrm{opt}\sim C_{\bar P}(1-\zeta)^{-2}$ for $i=1,2$ with $C_{\bar P}=2/(\sigma_\mathscr D^2\mathsf{acv}_0^2)$.
\end{lemma}
At low-illumination, from Lemma \ref{lem:barP_optimal_samples}, the optimal sample sizes for $\bar P$ are asymptotically proportional to $(1-\zeta)^{-2}$, demonstrating that as expected, infinite sample sizes are needed in the zero illumination limit.


\subsection{Statistical analysis of hat-{\normalfont $P^{-1}_\delta$}}
\label{subsec:analysis_of_hatP_inv_delta}

We now carry out a similar but more detailed analysis for $\hat P^{-1}_\delta$. Under the normal model of Section \ref{sec:sensor_noise_model} we have the distributional results $\hat X\sim\mathcal G(\alpha_1,\beta_1)$ and $\hat Y\sim\mathcal G(\alpha_2,\beta_2)$ (see Table \ref{tab:GEN_symbols}). As such, the difference $\hat P=\hat X-\hat Y$ is distributed as $\hat P\sim\mathcal{GD}(\alpha_1,\alpha_2,\beta_1,\beta_2)$, which is a gamma-difference variable. Working with the properties of the gamma-difference distribution as well as Lemma \ref{lem:parameter_relations} we can write
\begin{equation}
\label{eq:ACV2_hatPinv_delta}
\mathsf{ACV}^2\hat P^{-1}_\delta=\frac{2}{(1-\zeta)^2}\left(\frac{1}{n_1-1}+\frac{\zeta^2}{n_2-1}\right),
\end{equation}
which leads to the following expressions for optimal sample sizes of $\hat P^{-1}_\delta$.
\begin{lemma}
\label{lem:hatPinv_delta_optimal_samples}
The optimal sample sizes for $\hat P^{-1}_\delta$ are given by,
\begin{equation}
(n_1^\mathrm{opt},n_2^\mathrm{opt})=\left(\frac{2(1+\zeta)}{\mathsf{acv}_0^2(1-\zeta)^2}+1,\frac{2\zeta(1+\zeta)}{\mathsf{acv}_0^2(1-\zeta)^2}+1\right).
\end{equation}
\end{lemma}

\begin{remark}
Our goal for deriving the optimal sample sizes of $\hat P^{-1}_\delta$ was for the purpose of optimally estimating $\hat P^{-1}$. It turns out that we can alter the sample sizes of Lemma \ref{lem:hatPinv_delta_optimal_samples} to make them exact in the shot-noise-limit. See Remark \ref{rem:tuned_optimal_samples} in the appendix for more information.
\end{remark}

\begin{proposition}
\label{prop:optimal_sample_properties}
Let $n_1^\mathrm{opt}$ and $n_2^\mathrm{opt}$ denote the optimal sample sizes for $\hat P^{-1}_\delta$, then:
\begin{enumerate}
\item[(i)] $n_1^\mathrm{opt}$ and $n_2^\mathrm{opt}$ are strictly increasing with decreasing illumination ($\zeta$ increasing),
\item[(ii)] $n_2^\mathrm{opt}<n_1^\mathrm{opt}$ for all $\zeta\in[0,1)$,
\item[(iii)] As illumination decreases, $\zeta\to 1^-$, $n_2^\mathrm{opt}/n_1^\mathrm{opt}\to 1^-$, and $n_i^\mathrm{opt}\sim C_{\hat P^{-1}_\delta}(1-\zeta)^{-2}$ for $i=1,2$ where $C_{\hat P^{-1}_\delta}=4/\mathsf{acv}_0^2$.
\end{enumerate}
\end{proposition}

\begin{corollary}
\label{cor:minimum_possible_sample_sizes}
Let $n_1^\mathrm{opt}(\zeta)$ and $n_2^\mathrm{opt}(\zeta)$ denote the optimal sample sizes for $\hat P^{-1}_\delta$ as a function of $\zeta$. To estimate $\hat P^{-1}_\delta$ to a relative uncertainty of $\mathsf{acv}_0$, one needs at minimum a total of $N^\mathrm{opt}(0)=n_1^\mathrm{opt}(0)+n_2^\mathrm{opt}(0)=2/\mathsf{acv}_0+2$ observations.
\end{corollary}

Proposition \ref{prop:optimal_sample_properties} demonstrates that the optimal sample sizes for $\hat P^{-1}_\delta$ have similar characteristics for those of $\bar P$. In particular, we see at low-illumination that the optimal sample sizes are asymptotically equal and asymptotically proportional to $(1-\zeta)^{-2}$ showing that infinite sample sizes are again needed in the zero illumination limit. Furthermore, Corollary \ref{cor:minimum_possible_sample_sizes} takes advantage of the monotonicity of the optimal sample sizes in $\zeta$ in order to derive a lower bound to the minimum total number of samples needed to achieve an absolute coefficient of variation for $\hat P^{-1}_\delta$ equal to $\mathsf{acv}_0$. Note that this lower bound closely matches the asymptotic result for shot-noise-limited estimator of $g$ given in (\ref{eq:shot_noise_limited_optimal_samples}). Choosing $\mathsf{acv}_0=0.01$ gives the lower bound $N^\mathrm{opt}(0)=20\,002$, which serves to show the large sample sizes needed to obtain typically desired measurement uncertainties.

Now that we have explicit expressions for the optimal sample sizes of $\hat P^{-1}_\delta$, and know their properties, we can assess the accuracy of the approximation obtained by using the optimal sample sizes of $\hat P^{-1}_\delta$ in place of those for $\hat P^{-1}$.

\begin{theorem}
\label{thm:hatPinv_opt_ARB_ACV}
Let $\hat P^{-1}_\mathrm{opt}$ denote the estimator $\hat P^{-1}$ as a function of the optimal sample sizes for $\hat P^{-1}_\delta$. Then as $\mathsf{acv}_0\to 0^+$
\begin{equation}
\begin{aligned}
\mathsf{ACV}_\mathcal P\hat P^{-1}_\mathrm{opt} &=\mathsf{acv}_0+3\mathsf{acv}_0^3+\mathcal O(\mathsf{acv}_0^5),\\
\mathsf{ARB}_\mathcal P\hat P^{-1}_\mathrm{opt} &=\mathsf{acv}_0^2+3\mathsf{acv}_0^4+\mathcal O(\mathsf{acv}_0^6).
\end{aligned}
\end{equation}
\end{theorem}

As Theorem \ref{thm:hatPinv_opt_ARB_ACV} shows, when  $\mathsf{acv}_0$ is small, the optimal samples sizes for $\hat P^{-1}_\delta$ force $\mathsf{ACV}_\mathcal P\hat P^{-1}$ to be nearly $\mathsf{acv}_0$, which indicates a good approximation. Since the secondary term in the expansion for $\mathsf{ACV}_\mathcal P\hat P^{-1}_\mathrm{opt}$ is cubic, we expect the approximation to be quite good for any small choice of $\mathsf{acv}_0$, say $\mathsf{acv}_0<0.1$. This restriction is not problematic as one generally does not seek to measure $1/\sigma_P^2$ with greater than $10\%$ uncertainty. Furthermore, Theorem \ref{thm:hatPinv_opt_ARB_ACV} also quantifies the bias of $\hat P^{-1}_\mathrm{opt}$ showing that the main term is quadratic in $\mathsf{acv}_0$. As such, we expect the estimator $\hat P^{-1}$ to be nearly unbiased for $1/\sigma_P^2$ (in the sense that $\mathsf E_\mathcal P\hat P^{-1}\approx 1/\sigma_P^2$) when subject to the optimal sample sizes for $\hat P^{-1}_\delta$ for a small choice of $\mathsf{acv}_0$.


\subsection{Statistical analysis of {\normalfont $G_\delta$}}
\label{subsec:analysis_of_G_delta}

Equipped with the results of Subsections \ref{subsec:analysis_of_barP}-\ref{subsec:analysis_of_hatP_inv_delta}, we are now ready to perform a statistical analysis on $G_\delta$ with the goal of determining when the optimal sample sizes for $\hat P^{-1}_\delta$ serve as good approximations for those of $G_\delta$. Using the formula in Lemma \ref{lem:ACV2_of_product} we first write the squared absolute coefficient of variation for $G_\delta$ as
\begin{equation}
\mathsf{ACV}^2G_\delta=\mathsf{ACV}^2\bar P+(\mathsf{ACV}^2\bar P)(\mathsf{ACV}^2\hat P^{-1}_\delta)+\mathsf{ACV}^2\hat P^{-1}_\delta,
\end{equation}
with $\mathsf{ACV}^2\bar P$ and $\mathsf{ACV}^2\hat P^{-1}_\delta$ given in (\ref{eq:ACV2_barP}) and (\ref{eq:ACV2_hatPinv_delta}), respectively. Our first goal here is to show the estimator $G_\delta$ behaves like previous estimators for $g$ in the sense that its uncertainty is dominated by that of $\hat P^{-1}_\delta$ at high-illumination (c.f.~Theorem \ref{thm:ACVG_shot_limited_approx}).
\begin{theorem}
\label{thm:ACVG_general_approx}
Let $G_\delta$ and $\hat P^{-1}_\delta$ be as given in (\ref{eq:G_delta_estimator}) and (\ref{eq:Pinv_delta_estimator}), respectively. As illumination increases, $\zeta\to 0^+$ and
\begin{equation}
\mathsf{ACV}G_\delta=\mathsf{ACV}\hat P^{-1}_\delta\left(1+\frac{1}{\sigma_\mathscr D^2}\frac{n_1+1}{4n_1}\zeta+\mathcal O(\zeta^2)\right).
\end{equation}
\end{theorem}

Since $\mathsf{ACV}G_\delta\sim\mathsf{ACV}\hat P^{-1}_\delta$ we have by the same reasoning in Section \ref{sec:shot_noise_limited} that the optimal sample sizes for $\hat P^{-1}_\delta$ are asymptotically equal to those of $G_\delta$ at high-illumination with equality in the shot-noise-limit $(\zeta=0)$. However, if our goal is optimal sampling for the general case, this approximation must also hold in the low-illumination regime $(\zeta\to 1^-)$. To determine how good this approximation is at low-illumination, we may consider the following theorem.
\begin{theorem}
\label{thm:G_delta_asymptotic_optimal_samples}
Let $n_1^\mathrm{opt}$ and $n_2^\mathrm{opt}$ denote the optimal sample sizes for $G_\delta$. Then, as illumination decreases, $\zeta\to 1^-$, $n_2^\mathrm{opt}/n_1^\mathrm{opt}\to 1^-$, and $n_i^\mathrm{opt}\sim C_{G_\delta}(1-\zeta)^{-2}$ for $i=1,2$ with
\begin{equation}
C_{G_\delta}=\frac{2}{\mathsf{acv}_0^2}\left(1+\frac{1}{2\sigma_\mathscr D^2}+\left(\left(1+\frac{1}{2\sigma_\mathscr D^2}\right)^2+\frac{2}{\sigma_\mathscr D^2}\mathsf{acv}_0^2\right)^{1/2}\right).
\end{equation}
\end{theorem}

Comparing Theorem \ref{thm:G_delta_asymptotic_optimal_samples} with Proposition \ref{prop:optimal_sample_properties} $(iii)$ shows that at low-illumination the optimal sample sizes for $\hat P^{-1}_\delta$ are asymptotically proportional to the optimal sample sizes for $G_\delta$. The constant of proportionality, $C_{G_\delta}/C_{\hat P^{-1}_\delta}$, depends on the dark noise $\sigma_\mathscr D$ and as the dark noise increases
\begin{equation}
\frac{C_{G_\delta}}{C_{\hat P^{-1}_\delta}}\sim 1+\frac{1+\mathsf{acv}_0^2}{2\sigma_\mathscr D^2}+\mathcal O(\sigma_\mathscr D^{-4}),
\end{equation}
which further shows that this constant of proportionality approaches one with increasing dark noise. Bringing the observations following Theorems \ref{thm:ACVG_general_approx}-\ref{thm:G_delta_asymptotic_optimal_samples} together, we expect the optimal sample sizes of $\hat P^{-1}_\delta$ to serve as good approximations to those of $G_\delta$ at any illumination level given enough dark noise.

To get a better grasp on these observations we will define the metric
\begin{equation}
\mathcal E=\frac{\mathsf{ACV}\hat P^{-1}_\delta}{\mathsf{ACV}G_\delta}.
\end{equation}
and consider this metric as a function of the optimal sample sizes for $\hat P^{-1}_\delta$, that is,
\begin{equation}
\label{eq:Eopt_metric}
\mathcal E_\mathrm{opt}=\frac{\mathsf{acv}_0}{\mathsf{ACV}G_{\delta,\mathrm{opt}}}=\left(1+(\mathsf{ACV}^2\bar P_\mathrm{opt})(1+\mathsf{acv}_0^{-2})\right)^{-1/2}
\end{equation}
with
\begin{equation}
    \mathsf{ACV}^2\bar P_\mathrm{opt}=\frac{1}{\sigma_\mathscr D^2}\frac{\zeta}{(1-\zeta)^2}\left(\frac{1}{n_1^\mathrm{opt}}+\frac{\zeta}{n_2^\mathrm{opt}}\right).
\end{equation}
Both $\mathcal E$ and it's counterpart $\mathcal E_\mathrm{opt}$ are normalized in the sense that $0\leq\mathcal E_\mathrm{opt}\leq 1$. In particular, whenever $\mathcal E_\mathrm{opt}\approx 1$ we can conclude that the optimal sample sizes of $\hat P^{-1}_\delta$ are good approximations for those of $G_\delta$ with equality achieved when $\mathcal E_\mathrm{opt}=1$.

It is straightforward to show $\lim_{\zeta\to 0^+}\mathcal E_\mathrm{opt}=1$, which again indicates that the optimal sample sizes for $\hat P^{-1}_\delta$ are exactly equal to those for $G_\delta$ in the shot-noise-limit. Taking the limit now in the opposite direction, we see $\mathsf{ACV}^2\bar P_{\mathrm{opt}}\to\mathsf{acv}_0^2/(2\sigma_\mathscr D^2)$ as $\zeta\to 1^-$ and so
\begin{equation}
\overline{\mathcal E_\mathrm{opt}}\coloneqq\lim_{\zeta\to 1^-}\mathcal E_\mathrm{opt}=\left(1+\frac{1+\mathsf{acv}_0^2}{2\sigma_\mathscr D^2}\right)^{-1/2},
\end{equation}
which is nonzero. To get a sense of how close $\overline{\mathcal E_\mathrm{opt}}$ is to one recall from the discussion following Theorem \ref{thm:hatPinv_opt_ARB_ACV}, that we generally want to impose the restriction $\mathsf{acv}_0\in(0,0.1)$. Since $\overline{\mathcal E_\mathrm{opt}}$ is decreasing in $\mathsf{acv}_0$ we set $\mathsf{acv}_0=0.1$ giving the lower bound
\begin{equation}
\overline{\mathcal E_\mathrm{opt}}>\left(1+\frac{101}{200}\frac{1}{\sigma_\mathscr D^2}\right)^{-1/2}.
\end{equation}
Plotting this lower bound as a function of $\sigma_\mathscr D$ we find $\overline{\mathcal E_\mathrm{opt}}>0.99$ for $\sigma_\mathscr D\geq 5\,e\text{-}$ and so for sensors with dark noise greater than $5\,e\text{-}$, we can expect the optimal sample sizes of $\hat P^{-1}_\delta$ to give excellent approximations to those of $G_\delta$ at any level of illumination. In conclusion, if $\mathsf{acv}_0$ is small $(\mathsf{acv}_0<0.1)$ and the dark noise is sufficiently large $(\sigma_\mathscr D\geq 5\,e\text{-})$, the optimal sample sizes for $\hat P^{-1}_\delta$ serve as excellent approximations for those of the general conversion gain estimator $G$ at any illumination level.

We close this section with an analogous result to that of Theorem \ref{thm:hatPinv_opt_ARB_ACV}, which confirms our findings.
\begin{theorem}
\label{thm:G_opt_ARB_ACV}
Let $G_\mathrm{opt}$ denote the estimator $G$ as a function of the optimal sample sizes for $\hat P^{-1}_\delta$. Then as $\mathsf{acv}_0\to 0^+$
\begin{equation}
\begin{aligned}
\mathsf{ACV}_\mathcal PG_\mathrm{opt} &=\sqrt{1+a}\,\mathsf{acv}_0+\frac{6+a+b}{2\sqrt{1+a}}\mathsf{acv}_0^3+\mathcal O(\mathsf{acv}_0^5),\\
\mathsf{ARB}_\mathcal PG_\mathrm{opt} &=\mathsf{acv}_0^2+3\mathsf{acv}_0^4+\mathcal O(\mathsf{acv}_0^6),
\end{aligned}
\end{equation}
where $a=\frac{1}{\sigma_\mathscr D^2}\frac{\zeta}{1+\zeta}$ and $b=-\frac{1}{\sigma_\mathscr D^2}\frac{(1-\zeta)^2}{4(1+\zeta)}$.
\end{theorem}
Notice that in particular, if $\sigma_\mathscr D$ is large, and/or $\zeta$ is small, the main term in the expansion for $\mathsf{ACV}_\mathcal PG_\mathrm{opt}$ is close to $\mathsf{acv}_0$ showing that the optimal sample sizes for $\hat P^{-1}_\delta$ are good approximations for those of $G$ under these conditions.


\section{Design and control of experiment for per-pixel conversion gain estimation}
\label{sec:DOE_COE}

In this section, we demonstrate how the derived expressions can be used in the Design \& Control of Experiment ({\sc doe} \& {\sc coe}) for per-pixel conversion gain estimation.  For this example we will use the ON Semiconductor KAI-04070 monochrome interline transfer {\sc ccd} sensor \cite{ccd_sensor}. As we will demonstrate, the entire experimental process for per-pixel conversion gain estimation centers around the optimal sample sizes for $\hat P^{-1}_\delta$ in Lemma \ref{lem:hatPinv_delta_optimal_samples} and the metric $\mathcal E_\mathrm{opt}$ given in (\ref{eq:Eopt_metric}).

A full Monte Carlo implementation of the {\sc doe} and {\sc coe} algorithms presented in this section can be found in Code 1 \cite{supp_code1} and Code 2 \cite{supp_code2}. Code 1 (\texttt{DOE\_and\_COE.m}) is the primary script and requires {\sc matlab}'s \emph{Statistics and Machine Learning} toolbox to run.  Since Code 1 is meant to guide the reader in understanding the steps involved in {\sc doe} and {\sc coe} it is advised to run it one section at a time and observe the results of each section.


\subsection{Experimental Setup}
\label{subsec:experimental setup}

Figure \ref{fig:experimental_setup} shows a schematic diagram of the experimental setup used for performing per-pixel conversion gain estimation. The experimental setup consisted of a $650\,\mathrm{nm}$ Superluminescent Light Emitting Diode ({\sc sled}) passed into an integrating sphere with the Sensor Under Test ({\sc sut}) placed in the plane of the spheres output port ($f/0$ geometry) where uniformity is highest \cite{LABSPHERE_2017}. To facilitate control of the illumination level, a Variable Optical Attenuator ({\sc voa}) was introduced between the {\sc sled} and integrating sphere. The sensor was configured at its full bit-depth of $14$-bits to minimize quantization error and image data was read off the sensor at a $512\times 512\,\mathrm{px}$ resolution using a single readout register operating at $40\,\mathrm{mhz}$. By operating the sensor in a single-tap mode like this, the gain of each pixel will be the same and so any excursions in the measured gain of each pixel will be, in theory, completely due to sampling error.  This high level of uniformity will allow us to see if the proposed algorithms are able to successfully measure the gain to the desired relative uncertainty, verify our noise model, and compare experimental results to theoretical predictions.

In order to capture imagery under both dark and illuminated conditions, a Motorized Mirror ({\sc mm}) was placed next to the path of the {\sc sled} beam. Moving the mirror into the beam path redirected the beam away from the integrating sphere and into a Beam Dump ({\sc bd}); thus, providing a dark environment for the sensor.
\begin{figure}[htb]
\centering
\includegraphics[scale=0.25]{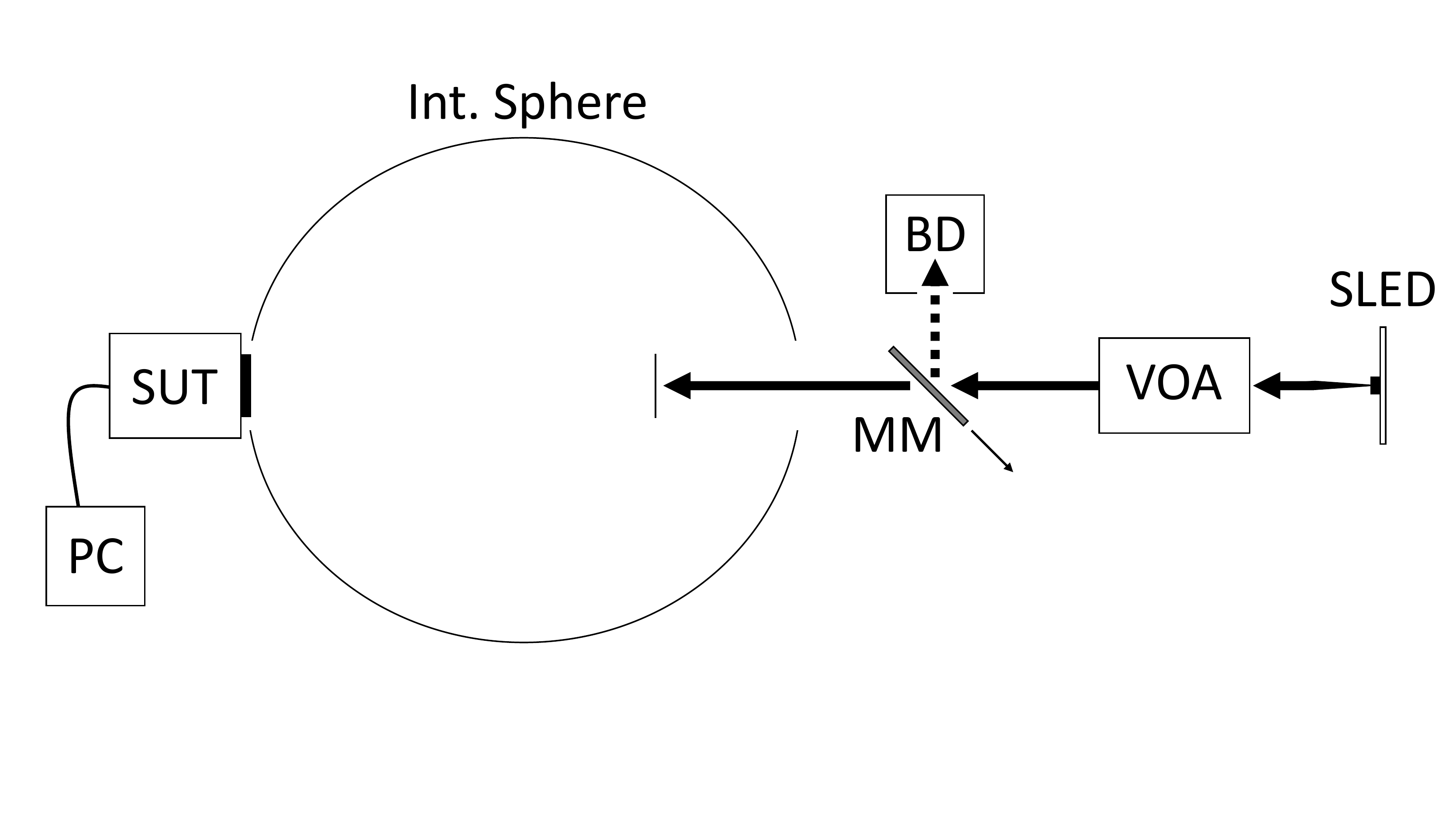}
\caption{Schematic diagram of experimental setup for per-pixel conversion gain estimation.}
\label{fig:experimental_setup}
\end{figure}


\subsection{Design of experiment}
\label{subsec:doe}

{\sc doe} for per-pixel conversion gain estimation begins with choosing a suitable value for the desired relative uncertainty $\mathsf{acv}_0$. As a rough rule of thumb, if we are operating under the condition $\mathcal E_\mathrm{opt}\approx 1$, then for small $\mathsf{acv}_0$ we may approximate $G_\mathrm{opt}\sim\mathcal N(g,(g\times\mathsf{acv}_0)^2)$ so that $\mathsf P(G_\mathrm{opt}\in g(1\pm\mathsf{acv}_0))=0.683$. This amounts to the optimal estimate of $g$ being within $\mathsf{acv}_0\times 100\%$ of its exact value $68.3\%$ of the time. For this particular experiment it was decided to use $\mathsf{acv}_0=0.05$ as this represents a typical value one might choose.

Next, we want to determine an appropriate illumination level for measuring the conversion gain by selecting a value of $\zeta$ where: (1) the total number of required observations is attainable and (2) our approximate optimal sample sizes are valid.  To determine when both conditions are satisfied, we will create an $\mathcal E$-$N$ plot, which consists of plotting the functions $\mathcal E_\mathrm{opt}$ and $N^\mathrm{opt}=n_1^\mathrm{opt}+n_2^\mathrm{opt}$ as a function of $\zeta$.  To determine an illumination level, we select a value of $\zeta$ where $N^\mathrm{opt}$ is small enough and $\mathcal E_{\mathrm{opt}}\approx 1$. However, because $\mathcal E_\mathrm{opt}$ contains the unknown value of the dark noise, we will need to provide an estimate of $\sigma_\mathscr D$ for this plot, which can be accomplished through a preliminary measurement, a vendor specification sheet, or an educated guess/lower bound.

In this example we will take the route of estimating a global lower bound for the dark noise. For this particular sensor we expect $g\geq 1$ and since $\sigma_\mathscr D=\sigma_D\times g$ it follows that $\sigma_D$ provides a lower bound on $\sigma_\mathscr D$. Since $\mathcal E_\mathrm{opt}$ is an increasing function of $\sigma_\mathscr D$, it follows that evaluating it at $\sigma_D$ provides a lower bound for it's exact value. To obtain a global estimate of $\sigma_D^2$ we capture two dark frames $\mathbf Y_1$ and $\mathbf Y_2$ and then compute half the sample variance of the difference-frame $\Delta \mathbf Y=\mathbf Y_1-\mathbf Y_2$. Note the use of bold symbols to denote $\mathbf Y_k$ as a two-dimensional array so that $\Delta \mathbf Y$ is the pixel-wise difference of the two dark frames $\mathbf Y_1$ and $\mathbf Y_2$. For the sensor under test we found
\begin{equation}
\label{eq:global_est_of_sD_elec}
\hat\sigma_D^2=\frac{1}{2}\mathrm{var}(\Delta\mathbf Y)=39.94\,\mathrm{DN}^2,
\end{equation}
where
\begin{equation}
\mathrm{var}(\Delta\mathbf Y)=\frac{1}{512^2-1}\sum_{i=1}^{512}\sum_{j=1}^{512}(\Delta \mathbf Y_{ij}-\overline{\Delta \mathbf Y})^2.
\end{equation}

Figure \ref{fig:DOE_plots} plots
\begin{multline}
\hat{\mathcal E}_\mathrm{opt}(\zeta)=\\
\left(1+\frac{1+\mathsf{acv}_0^{-2}}{\hat\sigma_D^2}\frac{\zeta}{(1-\zeta)^2}\left(\frac{1}{n_1^\mathrm{opt}(\zeta)}+\frac{\zeta}{n_2^\mathrm{opt}(\zeta)}\right)\right)^{-1}
\end{multline}
along with $N^\mathrm{opt}(\zeta)$ for our choice of $\mathsf{acv}_0=0.05$ and estimate $\hat\sigma_D$.
\begin{figure}[htb]
\centering
\includegraphics[scale=1]{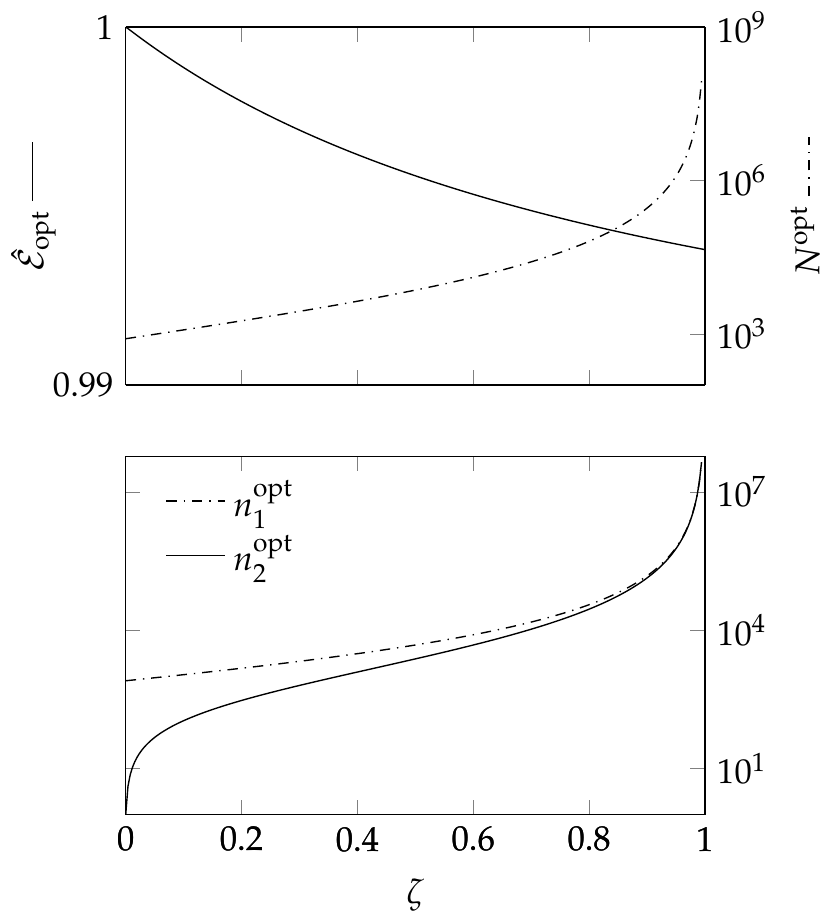}
\caption{$\mathcal E$-$N$ plot (top) with optimal sample sizes (bottom) for $\mathsf{acv}_0=0.05$ versus $\zeta$.}
\label{fig:DOE_plots}
\end{figure}
Due to the sufficiently large value of $\hat\sigma_D$ we see that $\hat{\mathcal E}_\mathrm{opt}$ is near unity for virtually any illumination level so that the requirement $\mathcal E_\mathrm{opt}\approx 1$ will not restrict what illumination levels we can choose for the experiment. To select an appropriate illumination level we first note that this sensor can record images at $\approx 5\,\mathrm{fps}$ for the chosen readout rate of $40\,\mathrm{mhz}$. Looking back at Figure \ref{fig:DOE_plots} we observe that the illumination level corresponding to $\zeta= 0.4$ is paired with an optimal total sample size of $N^\mathrm{opt}\approx 4300$ and approximation quality metric $\mathcal E_\mathrm{opt}\approx 0.993$. At a recording rate of $5\,\mathrm{fps}$ this number of images will take $\approx 15\,\mathrm{min.}$ to capture, which is short enough to avoid any significant drift in the sensor or source.

Now equipped with a desired value for $\zeta$, we guessed the required illumination level by adjusting the variable optical attenuator until the illumination level at the sensor plane resulted in a mean pixel output of about $1\%$ of the sensor's dynamic range. Using the same process to estimate the dark noise, two illuminated frames $\mathbf X_1$ and $\mathbf X_2$ were captured and their difference $\Delta\mathbf X=\mathbf X_1-\mathbf X_2$ was used to obtain the estimate
\begin{equation}
\label{eq:global_est_of_sPD_elec}
\hat\sigma_{P+D}^2=\frac{1}{2}\mathrm{var}(\Delta\mathbf X)=112.89\,\mathrm{DN}^2,
\end{equation}
which lead to a global estimate of $\zeta$ equal to
\begin{equation}
\label{eq:pre_experiment_zeta_est}
\hat\zeta=\frac{\hat\sigma_D^2}{\hat\sigma_{P+D}^2}=0.354.
\end{equation}
For our purposes, this illumination level was sufficiently close to the target $\zeta=0.4$ and corresponded to $\hat{\mathcal E}_\mathrm{opt}=0.996$ and $N^\mathrm{opt}=3521$ images, which needs only $\approx 12\,\mathrm{min}.$ to capture. If this $\zeta$-value was not appropriate we could simply re-adjust the illumination level, capture another two illuminated frames, and re-estimate $\zeta$.  We can repeat this procedure until we have found an illumination level corresponding to a suitable $\zeta$-value.

Before moving on to data capture, its important to understand what (\ref{eq:pre_experiment_zeta_est}) is an estimate of. Recall that $\mathbf Y_{ijk}\sim\mathcal N(\mu_D,\sigma_D^2)$ so that we have for the difference frame $\mathbf Y_{ij1}-\mathbf Y_{ij2}=\Delta\mathbf Y_{ij}\sim\mathcal N(0,2\sigma_D^2)$. In general $\sigma_D^2$ will vary from pixel-to-pixel (the $ij$-dimension) so that we may treat it as a random variable and model it according to some probability distribution $\sigma_D^2\sim F_{\sigma_D^2}$. By the law of total variance one finds
\begin{equation}
    \mathsf{Var}(\Delta\mathbf Y_{ij})=\mathsf{Var}(\mathsf E(\Delta\mathbf Y_{ij}|\sigma_D^2))+\mathsf E(\mathsf {Var}(\Delta\mathbf Y_{ij}|\sigma_D^2))=2\,\mathsf E(\sigma_D^2),
\end{equation}
which shows that the estimator (\ref{eq:global_est_of_sD_elec}) gives an unbiased estimate of the average value of $\sigma_D^2$ across the sensor array and likewise for the estimator of $\sigma_{P+D}^2$ in (\ref{eq:global_est_of_sPD_elec}). Thus, $\hat\zeta$ is a ratio of unbiased estimates for the average noise values across the sensor array; providing a useful global estimate of $\zeta$.


\subsection{Control of experiment}
\label{subsec:coe}

With the illumination level set, the experiment was ready to commence. The algorithm for data capture is presented in Algorithm \ref{alg:coe}. This algorithm utilizes Welford's online algorithm in {\sc UpdateStats}() to iteratively update the master frames $\bar{\mathbf X}$, $\bar{\mathbf Y}$, $\hat{\mathbf X}$, and $\hat{\mathbf Y}$, which contain sample means and variances for each pixel in both dark and illuminated conditions. In each iteration of the algorithm a batch of illuminated frames and another batch of dark frames are captured and used to update the master frames before recalculating a global estimate of $\zeta$, which is then in turn used to update the next batch sizes.  The initial batch sizes are $\mathtt{batch1}=\lceil n_1^\mathrm{opt}(0,\mathsf{acv}_0)\rceil$ and $\mathtt{batch2}=2$, which are the minimal possible number of each frame type needed. The batch sizes are updated to capture $m\times 100\%$ ($m\in(0,1]$) of the remaining difference between the current sample sizes and their estimates.  The parameter $m$ controls how aggressive the algorithm is and ultimately how many times the light source needs to be turned on and off. Large values of $m$ mean the algorithm will iterate less times (putting more confidence in the estimated sample sizes) while smaller values of $m$ result in more iterations (less confidence in the estimates). In the limit $m\to 0$, the algorithm iterates every time a dark and illuminated frame are captured. The algorithm terminates when both batch sizes are nonpositive indicating that the current sample sizes meet or exceed their respective estimates. This is then followed by the per-pixel calculation of $g$, denoted $\mathbf G$, which we shall call the $g$-map.

While the form of $\hat\zeta$ in Algorithm \ref{alg:coe} has a much different form than that of (\ref{eq:pre_experiment_zeta_est}), we can show it still is estimating the same quantity. To see why note that $\hat{\mathbf Y}_{ij}\sim\mathcal G(\alpha_2,\beta_2)$ with $\sigma_D^2\sim F_{\sigma_D^2}$. By the law of total expectation
\begin{equation}
\mathsf E(\hat{\mathbf Y}_{ij})=\mathsf E(\mathsf E(\hat{\mathbf Y}_{ij}|\sigma_D^2))=\mathsf E(\sigma_D^2),
\end{equation}
which shows $\sum\hat{\mathbf Y}_{ij}$ is an unbiased estimator for $I\cdot J\cdot\mathsf E(\sigma_D^2)$ with $I$ and $J$ being the vertical and horizontal resolution of the sensor in units of pixels, respectively. Likewise, $\sum\hat{\mathbf X}_{ij}$ is an unbiased estimator for $I\cdot J\cdot \mathsf E(\sigma_{P+D}^2)$ so that $\hat\zeta=\sum\hat{\mathbf Y}_{ij}/\sum\hat{\mathbf X}_{ij}$ is equivalent to a ratio of unbiased estimates for $\mathsf E(\sigma_D^2)$ and $\mathsf E(\sigma_{P+D}^2)$ just as (\ref{eq:pre_experiment_zeta_est}) is.

\begin{algorithm}[htb]
\caption{{\sc coe} algorithm.}\label{alg:coe}
\begin{algorithmic}[1]
\Procedure{OptimalCOE}{$\mathsf{acv}_0$,$m$}
\State $n_1=0$; $n_2=0$; $\hat\zeta=0;$
\State $\mathtt{batch1}=\lceil n_1^\mathrm{opt}(\hat\zeta,\mathsf{acv}_0)\rceil$;
\State $\mathtt{batch2}=2$; \label{line:n2_init_batch}
\State
\While{$\mathtt{batch1}>0\lor\mathtt{batch2}>0$}
\If{$\mathtt{batch1}>0$}
\State Turn on light source.
\State $\mathtt{NewData}=$ Capture $\mathtt{batch1}$ $\mathbf X$-frame(s); \label{line:n1_batch_update}
\State $n_1=n_1+\mathtt{batch1}$;
\State $[\bar{\mathbf X},\hat{\mathbf X}]=\text{\sc{UpdateStats}}(\bar{\mathbf X},\hat{\mathbf X},\mathtt{NewData})$;
\EndIf
\If{$\mathtt{batch2}>0$}
\State Turn off light source.
\State $\mathtt{NewData}=$ Capture $\mathtt{batch2}$ $\mathbf Y$-frame(s); \label{line:n2_batch_update}
\State $n_2=n_2+\mathtt{batch2}$;
\State $[\bar{\mathbf Y},\hat{\mathbf Y}]=\text{\sc{UpdateStats}}(\bar{\mathbf Y},\hat{\mathbf Y},\mathtt{NewData})$;
\EndIf
\State $\hat\zeta=\sum\hat{\mathbf Y}_{ij}/\sum\hat{\mathbf X}_{ij}$;
\State $\mathtt{batch1}=\lceil m\times(n_1^\mathrm{opt}(\hat\zeta,\mathsf{acv}_0)-n_1)\rceil$;
\State $\mathtt{batch2}=\lceil m\times(n_2^\mathrm{opt}(\hat\zeta,\mathsf{acv}_0)-n_2)\rceil$;
\EndWhile
\State
\State $\mathbf G=(\bar{\mathbf X}-\bar{\mathbf Y})/(\hat{\mathbf X}-\hat{\mathbf Y})$;
\Comment{Calculation is per-pixel.}
\State \textbf{return} $\mathbf G$
\EndProcedure
\end{algorithmic}
\end{algorithm}


\subsection{Experimental Data}
\label{subsec:experimental_data}

Algorithm \ref{alg:coe} was executed on the chosen {\sc ccd} for the specified uncertainty $\mathsf{acv}_0=0.05$ and multiplier $m=0.8$, which halted after capturing $n_1=2602$ illuminated frames and $n_2=921$ dark frames. Due to the large value of $m$, the mirror only needed to be moved eleven times throughout the entire {\sc coe} procedure. Figure \ref{fig:G_map} presents the final $g$-map created from Algorithm \ref{alg:coe} along with a histogram of the $g$-map values. Upon inspection of the $g$-map, there is no evidence of structure or nonuniformities as was expected.  Because the excursions in the $g$-map are almost entirely due to random sampling error, as opposed to nonuniformity in the sensor or source, and the sensor noise obeys the assumed noise model, the estimated value of $g$ for each pixel represents an i.i.d.~observation from the {\sc cif} distribution \cite{Hendrickson:17}. Given the large values of $n_1$ and $n_2$, the fitted {\sc cif} distribution for this example was approximated by Hinkley's normal ratio distribution \cite{hinkley_1969}
\begin{equation}
    f_G(g)\sim\frac{e^{-c/2}}{\sigma_{\bar P}\sigma_{\hat P} a^2(g)}\left(\frac{b(g)e^{b^2(g)/2}}{\sqrt{2\pi}}(2\Phi(b(g))-1)+\frac{1}{\pi}\right)
\end{equation}
where $a(g)=(g^2/\sigma_{\bar P}^2+1/\sigma_{\hat P}^2)^{1/2}$, $b(g)=(\mu_{\bar P} g/\sigma_{\bar P}^2+\mu_{\hat P}/\sigma_{\hat P}^2)/a(g)$, and $c=\mu_{\bar P}^2/\sigma_{\bar P}^2+\mu_{\hat P}^2/\sigma_{\hat P}^2$. The parameters $\mu_{\bar P}$, $\sigma_{\bar P}$, $\mu_{\hat P}$, and $\sigma_{\hat P}$ are easily estimated from sample means and sample standard deviations of the $\bar{\mathbf P}=\bar{\mathbf X}-\bar{\mathbf Y}$ and $\hat{\mathbf P}=\hat{\mathbf X}-\hat{\mathbf Y}$ master frames.

Given that the fitted distribution is that of a normal ratio, the pseudomoments of the fitted distribution are exactly described by (\ref{eq:higher_order_pseudo_moment_of_G_aprx}). Table \ref{tab:pseudomoment_comparison} presents the exact values of the pseudomoments for the fitted distribution and compares them to sample moments of the $g$-map.  Upon inspection, we see the theoretical pseudomoments show a high level of agreement with the sample moments and thus demonstrate that the pseudomoments serve as a useful way to characterize moments of actual sensor data.
\begin{table}[htb]
\centering
\caption{\bf Comparison of $g$-map sample moments to fitted distribution pseudomoments.}
\begin{tabular}{cccc}
\hline
&fit &sample &error ($\%$) \\
\hline
$\mathsf E_\mathcal PG$ &$2.1971$ &$2.1972$ &$4.3508\times 10^{-3}$ \\
$\sqrt{\mathsf{Var}_\mathcal PG}$ &$0.11069$ &$0.11122$ &$0.47739$ \\
\hline
\end{tabular}
 \label{tab:pseudomoment_comparison}
\end{table}

\begin{figure}[htb]
\centering
\includegraphics[scale=1]{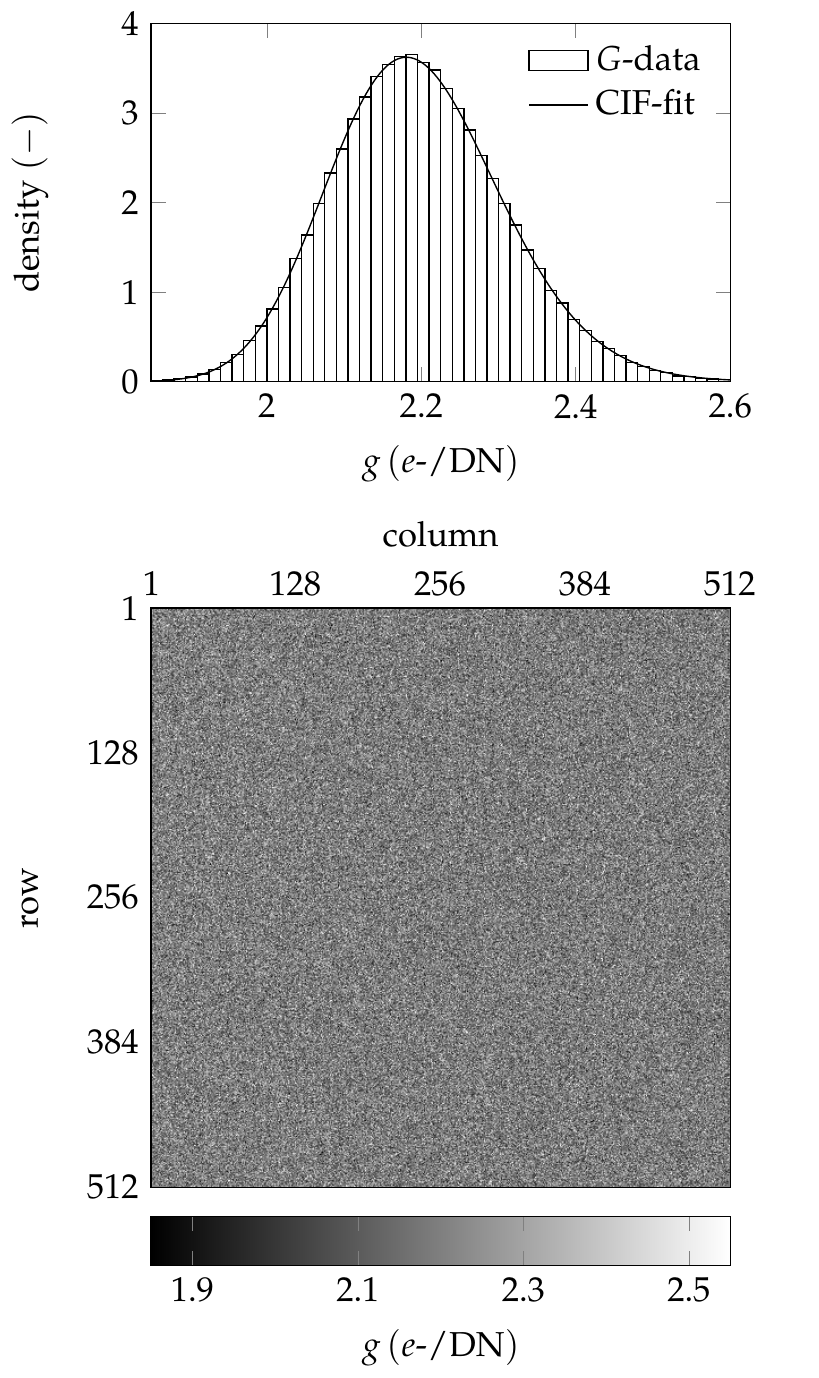}
\caption{Histogram of $g$-map values fit with the {\sc cif} distribution (top) and $g$-map (bottom).}
\label{fig:G_map}
\end{figure}

Recall that the {\sc coe} algorithm was run with a design parameter of $\mathsf{acv}_0=0.05$; comparing this to the measurement, we find the sample absolute coefficient of variation of the $g$-map to be
\begin{equation}
\mathrm{acv}(\mathbf G)=0.05059.
\end{equation}
Here, $\mathrm{acv}(\mathbf G)=\sqrt{\mathrm{var}(\mathbf G)}/|\mathrm{mean}(\mathbf G)|$ with $\mathrm{mean}(\mathbf G)$ and $\mathrm{var}(\mathbf G)$ denoting the sample mean and variance of the two-dimensional $g$-map, respectively. We note that $\mathrm{acv}(\mathbf G)>\mathsf{acv}_0$. The discrepancy between $\mathrm{acv}(\mathbf G)$ and $\mathsf{acv}_0$ is due in part to variance in the {\sc coe} procedure and other factors like having to estimate the optimal sample sizes. But even if the algorithm was deterministic and we knew the optimal sample sizes exactly, Theorem \ref{thm:G_opt_ARB_ACV} tells us that the absolute coefficient of variation for the $g$-map should have a positive bias away from the desired value of $\mathsf{acv}_0$.

To study this effect more, a Monte Carlo experiment was setup to replicate the sensor data presented above. The only parameters we need to estimate for the experiment are $g$, $\mu_\mathscr D$, $\sigma_\mathscr D$, and $\mu_{e\text{-}}$. To obtain a good estimate for $g$ we use the definition of relative bias $\mathsf{RB}_\mathcal PG=(\mathsf E_\mathcal PG-g)/g$ to write the first moment of $G$ as
\begin{equation}
    \mathsf E_\mathcal PG=g(1+\mathsf{RB}_\mathcal PG).
\end{equation}
In the context of conversion gain estimation we expect the relative bias to be positive, i.e. $\mathsf{RB}_\mathcal PG=\mathsf{ARB}_\mathcal PG$ and since the $g$-map, $\mathbf G$, was computed using the optimal sample sizes we have according to Theorem \ref{thm:G_opt_ARB_ACV}: $\mathsf{ARB}_\mathcal PG_\mathrm{opt}\sim\mathsf{acv}_0^2+3\mathsf{acv}_0^4$. This leads to the bias corrected estimate
\begin{equation}
    \hat g=\mathrm{mean}(\mathbf G)/(1+\mathsf{acv}_0^2+3\mathsf{acv}_0^4)
\end{equation}
with $\mathsf{acv}_0=0.05$. Table \ref{tab:MC_param_estimates} presents the estimate for $g$ along with the other estimates needed for the experiment. We note that these parameter estimates are used as the simulation parameters in Code 1 \cite{supp_code1}.
\begin{table}[htb]
\centering
\caption{\bf Parameter estimates for Monte Carlo experiment.}
\begin{tabular}{ccc}
\hline
parameter &estimate &value \\
\hline
$\hat g$ &$\mathrm{mean}(\mathbf G)/(1+\mathsf{acv}_0^2+3\mathsf{acv}_0^4)$ &$2.1917$ \\
$\hat\mu_\mathscr D$ &$\mathrm{mean}(\bar{\mathbf Y})\times \hat g$ &$92.858$ \\
$\hat\sigma_\mathscr D$ &$\sqrt{\mathrm{mean}(\hat{\mathbf Y})}\times \hat g$ &$13.853$ \\
$\hat\mu_{e\text{-}}$ &$\mathrm{mean}(\bar{\mathbf X}-\bar{\mathbf Y})\times \hat g$ &$350.03$ \\
\hline
\end{tabular}
 \label{tab:MC_param_estimates}
\end{table}

With the parameter estimates in Table \ref{tab:MC_param_estimates}, Algorithm \ref{alg:coe} was run a total of ten times using the design parameters $\mathsf{acv}_0=0.05$ and $m=0.8$. To simulate the sensor data, $512\times 512\mathrm{px}$ dark images, $\mathbf Y$, and illuminated images, $\mathbf X$, were generated according to the model in Section \ref{sec:sensor_noise_model} so that they contained elements of the form
\begin{gather}
    \mathbf Y_{ij}=\lceil\mathscr D/\hat g\rfloor\\
    \mathbf X_{ij}=\lceil(\mathscr P+\mathscr D)/\hat g\rfloor,
\end{gather}
with $\mathscr P\sim\mathcal P(\hat\mu_{e\text{-}})$ and $\mathscr D\sim\mathcal N(\hat\mu_\mathscr D,\hat\sigma_\mathscr D^2)$. Table \ref{tab:MC_results} presents the results of the Monte Carlo experiment. We see that the sample sizes resulting from the simulated {\sc coe} algorithm were very close to those from our actual experiment and have very little variance. This indicates that: (1) our assumed noise model is effective and (2) the difference between $\mathrm{acv}(\mathbf G)$ and $\mathsf{acv}_0$ in our original experiment is unlikely to be due to uncertainty in the optimal sample size estimates. Furthermore, we see that the absolute coefficient of variation for the $g$-map created by the algorithm shows a distinct positive bias away the desired value $\mathsf{acv}_0$, also with small variance. Such results are indicative that the difference between $\mathrm{acv}(\mathbf G)$ and $\mathsf{acv}_0$ in our original experiment is a result of approximating optimal sample sizes for $G$ with those of $\hat P^{-1}_\delta$.

\begin{table}[htb]
\centering
\caption{\bf Results of Monte Carlo experiment for ten runs.}
\begin{tabular}{cccc}
\hline
&$\mathrm{mean}(\cdot)$ &$\mathrm{var}(\cdot)$ &target \\
\hline
$n_1$ &$2606.9$ &$0.1$ &$2602$ \\
$n_2$ &$923.9$ &$0.1$ &$921$ \\
$\mathrm{acv}(\mathbf G)$ &$0.050279$ &$5.0328\times 10^{-9}$ &$0.050409$ (Thm. \ref{thm:G_opt_ARB_ACV}) \\
\hline
\end{tabular}
 \label{tab:MC_results}
\end{table}


\section{Application: Per-Pixel Read Noise Estimation}
\label{sec:R_map}

With a method for estimating per-pixel conversion gain, we can estimate downstream parameters such as the read noise $\sigma_\mathscr R$ on a per-pixel basis. Recall that the random variable $\mathscr D$ represents the dark noise of the sensor for some nonzero integration time. As such, the noise represented by $\mathscr D$ contains the combined noise from read noise and dark current shot noise. To isolate the sensor read noise component of $\mathscr D$, given by $\mathscr R\sim\mathcal N(\mu_\mathscr R,\sigma_\mathscr R^2)$, we simply set the sensor integration time to $t_\mathrm{exp}=0\,\mathrm{sec}$ (or the shortest allowable integration time), which eliminates dark current. Again by linearity of the transfer function we find for the sensor output at zero illumination for a zero second integration time to be $R=\mathcal T(\mathscr R)$ with
\begin{equation}
    \mu_R\coloneqq\mathsf ER=\mu_\mathscr R/g
\end{equation}
and
\begin{equation}
    \sigma_R^2\coloneqq\mathsf{Var}R=\sigma_\mathscr R^2/g^2.
\end{equation}
Upon inspection of the expression for $\sigma_R^2$ we find an equation for the read noise as 
\begin{equation}
    \sigma_\mathscr R=\sigma_R\times g.
\end{equation}
We already have an estimator for $g$ so all that is needed to estimate $\sigma_\mathscr R$ is an estimator for $\sigma_R$.

To estimate $\sigma_R$ we let $\{Z_1,\dots,Z_{n_3}\}$ denote a sequence of $n_3$ i.i.d.~observations of a pixel captured in the dark over a zero second integration time. Under the assumed normal model, $Z_k\sim\mathcal N(\mu_R,\sigma_R^2)$; thus we can obtain an unbiased estimate of $\sigma_R^2$ via the sample variance
\begin{equation}
    \hat Z=\frac{1}{n_3-1}\sum_{k=1}^{n_3}(Z_k-\bar Z)^2.
\end{equation}
It follows that we can estimate $\sigma_R$ via the sample standard deviation $\breve Z=\sqrt{\hat Z}$ although this estimate will no longer be unbiased. It's easy to show that $\hat Z\sim\mathcal G(\alpha_3,\beta_3)$ with $\alpha_3=(n_3-1)/2$ and $\beta_3=\alpha_3/\sigma_R^2$ so that using the transformation $\breve Z=\sqrt{\hat Z}$ we find for the density of $\breve Z$
\begin{equation}
f_{\breve Z}(z)=\frac{2\beta^\alpha}{\Gamma(\alpha)}z^{2\alpha-1}e^{-\beta z^2},
\end{equation}
which is the Nakagami distribution $\breve Z\sim{\mathcal Na}(\alpha_3,\beta_3)$. One finds for the moments of $\breve Z$
\begin{equation}
\mathsf E\breve Z^n=\frac{(\alpha_3)_{n/2}}{\alpha_3^{n/2}}\sigma_R^n,
\end{equation}
with $(s)_n\coloneqq\Gamma(s+n)/\Gamma(s)$ again denoting the Pochhammer symbol. Thus, an unbiased estimator for $\sigma_R$
can be given by
\begin{equation}
    \breve Z^\ast=\frac{\sqrt{\alpha_3}}{(\alpha_3)_{1/2}}\sqrt{\hat Z}.
\end{equation}
The read noise $\sigma_\mathscr R$ of our pixel can therefore be estimated via
\begin{equation}
    \breve{\mathscr R}=\breve Z^\ast\times G,
\end{equation}
with the estimator $G$ given by (\ref{eq:sub_shot_noise_limited_estimator}). Since $\breve Z^\ast$ is computed from a the sample $\mathbf Z$, which is independent from the samples $\mathbf X$ and $\mathbf Y$ used to calculate $G$, we further find for the density of $\breve R$ 
\begin{equation}
\label{eq:rn_distribution}
f_{\breve R}(r)=\int_0^\infty f_{\breve Z^\ast}(t)f_G(r/t)\frac{\mathrm dt}{t},
\end{equation}
with $f_G$ again being the {\sc cif} distribution. We shall simply refer to $f_{\breve{\mathscr R}}$ as the Read Noise ({\sc rn}) distribution. Likewise, we have for the pseudomoments of $\breve{\mathscr R}$
\begin{equation}
\mathsf E_\mathcal P\breve{\mathscr R}^n=\frac{(\alpha_3)_{n/2}}{(\alpha_3)_{1/2}^n}\sigma_R^n\,\mathsf E_\mathcal PG^n.
\end{equation}

To experimentally measure the per-pixel read noise of our sensor, we again employed an iterative algorithm.  In each iteration of this algorithm a master $\breve{\mathbf Z}^\ast$-frame was updated with a new $\mathbf Z$-frame using Welford's algorithm, which was then multiplied, per-pixel, by the $g$-map estimated in the previous section to produce the $\sigma_{\mathscr R}$-map.  The algorithm was stopped when the sample absolute coefficient of variation for the $\sigma_{\mathscr R}$-map satisfied $\mathrm{acv}(\breve{\mathscr R}_\mathrm{map})\leq 1.05\,\mathrm{acv}(\mathbf G)$. This algorithm halted after $n_3=2291$ $\mathbf Z$-frames were captured.

Figure \ref{fig:R_map} presents the $\sigma_{\mathscr R}$-map generated from this procedure along with its histogram fit the with {\sc rn}-distribution of (\ref{eq:rn_distribution}). Unlike the $g$-map, the $\sigma_{\mathscr R}$-map does show some column-wise nonuniformities, which are linked to the interline transfer CCD architecture of the sensor. However, these nonuniformities are not severe as indicated by how well the {\sc rn}-distribution fits the data. What this demonstrates is that per-pixel maps allow further insight into how sensor architecture affects the uniformity of key performance parameters across the sensor array.

\begin{figure}[htb]
\centering
\includegraphics[scale=1]{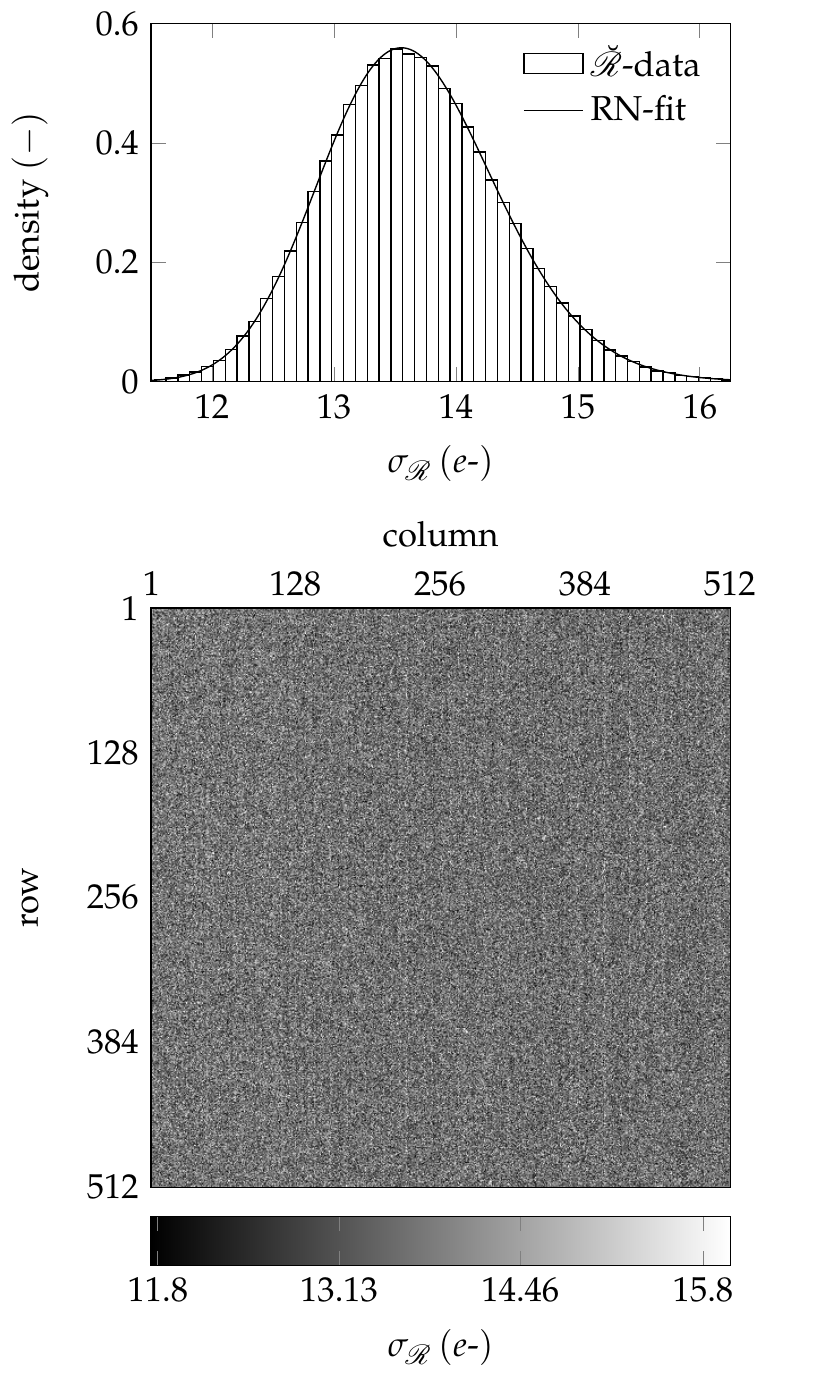}
\caption{Histogram of $\sigma_{\mathscr R}$-map values fit with the {\sc rn} distribution (top) and $\breve R$-map (bottom).}
\label{fig:R_map}
\end{figure}


\section{Conclusions}
\label{sec:conclusions}

In this work we have presented a general method for sample size determination of the photon transfer conversion gain measurement given a desired uncertainty requirement. So long as this uncertainty requirement is small and sensor dark noise is greater than $5\,e\text{-}$, this method of determining optimal sample sizes works across the full dynamic range of a sensor. Additionally, we have developed analytical expressions for the moments of the conversion gain sampling distribution (the {\sc cif} distribution) through the use of pseudomoments, showing that these pseudomoments accurately describe the sampling moments of conversion gain data under the proposed sensor noise model. With these theoretical results, we were able to construct simple design and control of experiment procedures that guides the number of samples required for both dark and illuminated conditions based on iterative statistics and predicted convergence. These experimental procedures were executed on a real image sensor; the results of which agreed with our theoretical predictions and were further confirmed through Monte Carlo simulation.

The ability to optimally measure per-pixel conversion gain is a key development in a more comprehensive approach to per-pixel photon transfer characterization.  We have already shown how per-pixel conversion gain maps enable per-pixel read noise estimation and we plan to extend this approach to measure other important {\sc pt} parameters such as well-capacity and dynamic range on a per-pixel basis. The additional information that comes from these per-pixel maps provides a richer characterization of the sensor and opens up the idea of assigning quality metrics to a sensor based on the uniformity of its per-pixel maps.


\section{Proofs}
\label{sec:proofs}

\begin{proof}[Proof of Lemma \ref{lem:ACV2_of_product}]
By the law of total variance
\begin{equation}
\mathsf{Var}T=(\mathsf EX^2)\mathsf{Var}Y+(\mathsf EY)^2\mathsf{Var}X.
\end{equation}
Substituting $\mathsf EX^2=\mathsf{Var}X+(\mathsf EX)^2$, expanding, and dividing both sides by $(\mathsf ET)^2=(\mathsf EX)^2(\mathsf EY)^2$ gives the desired result.
\end{proof}

\begin{proof}[Proof of Theorem \ref{thm:ACVG_shot_limited_approx}]
Letting $T=G$, $X=\bar P$, and $Y=\hat P^{-1}$ we have after dividing both sides of the relation in Lemma \ref{lem:ACV2_of_product} by $\mathsf{ACV}^2\hat P^{-1}$ and combining with
\begin{equation}
\mathsf{ACV}^2\bar P=\frac{1}{n}\frac{\sigma_P^2}{\mu_P^2}=\frac{1}{n}\frac{1}{\mu_Pg}=\frac{1}{n}\frac{1}{\mu_{e\text{-}}}
\end{equation}
and $\mathsf{ACV}^2\hat P^{-1}=2/(n-5)$:
\begin{equation}
\frac{\mathsf{ACV}G}{\mathsf{ACV}\hat P^{-1}}=\left(1+\frac{n-3}{2n}\frac{1}{\mu_{e\text{-}}}\right)^{1/2}.
\end{equation}
The result then follows from noting that $\sqrt{1+a/x}=1+a/(2x)+\mathcal O(x^{-2})$ as $x\to\infty$.
\end{proof}

\begin{proof}[Proof of Lemma \ref{lem:parameter_relations}]
The relation for $\sigma_P^2$ is derived by combining $\sigma_P^2=\sigma_{P+D}^2-\sigma_D^2$ and $\sigma_{P+D}^2=\sigma_D^2/\zeta$. To derive the relationship for $\mu_P$ we use the fundamental gain relation in (\ref{eq:modified_gain_relation}) to write $\mu_P=\sigma_P^2g$.
\end{proof}

\begin{proof}[Proof of Lemma \ref{lem:barP_optimal_samples}]
From (\ref{eq:ACV2_barP}) we observe $\mathsf{ACV}^2\bar P(N-n_2,n_2)$ is strictly convex on $n_2\in(0,N)$ for all $\zeta\in(0,1)$; thus, the optimal sample sizes for $\bar P$ satisfy
\begin{equation}
\begin{aligned}
\partial_{n_2} \mathsf{ACV}^2\bar P(N-n_2,n_2)\Big|_{N=n_1^{\mathrm{opt}}+n_2^{\mathrm{opt}},n_2=n_2^{\mathrm{opt}}} &=0,\\
\mathsf{ACV}^2\bar P(n_1^{\mathrm{opt}},n_2^{\mathrm{opt}}) &=\mathsf{acv}_0^2.
\end{aligned}
\end{equation}
The first equation gives the optimality relation $n_2^\mathrm{opt}/n_1^\mathrm{opt}=\sqrt\zeta$, which proves $n_2^\mathrm{opt}/n_1^\mathrm{opt}\to 1$ as $\zeta\to 1^-$. Solving the systems of equations then gives
\begin{equation}
(n_1^\mathrm{opt},n_2^\mathrm{opt})=\left(\frac{1}{\sigma_\mathscr D^2\mathsf{acv}_0^2}\frac{\zeta(1+\sqrt\zeta)}{(1-\zeta)^2},\frac{1}{\sigma_\mathscr D^2\mathsf{acv}_0^2}\frac{\zeta(\sqrt\zeta+\zeta)}{(1-\zeta)^2}\right).
\end{equation}
We now see that $n_2^\mathrm{opt}<n_1^\mathrm{opt}$, which allows us to deduce the stronger result $n_2^\mathrm{opt}/n_1^\mathrm{opt}\to 1^-$ as $\zeta\to 1^-$. Then working directly with the expressions for the optimal sample sizes it is straightforward to show $n_i^\mathrm{opt}\sim 2/(\sigma_\mathscr D^2\mathsf{acv}_0^2)(1-\zeta)^{-2}$, which completes the proof.
\end{proof}

\begin{proof}[Proof of Lemma \ref{lem:hatPinv_delta_optimal_samples}]
Upon inspection of (\ref{eq:ACV2_hatPinv_delta}) we see $\mathsf{ACV}^2\hat P^{-1}_\delta(N-n_2,n_2)$ is strictly convex on $n_2\in(1,N-1)$ for all $\zeta\in(0,1)$; thus, its optimal sample sizes satisfy
\begin{equation}
\begin{aligned}
\partial_{n_2} \mathsf{ACV}^2\hat P^{-1}_\delta(N-n_2,n_2)\Big|_{N=n_1^{\mathrm{opt}}+n_2^{\mathrm{opt}},\,n_2=n_2^{\mathrm{opt}}} &=0,\\
\mathsf{ACV}^2\hat P^{-1}_\delta(n_1^{\mathrm{opt}},n_2^{\mathrm{opt}}) &=\mathsf{acv}_0^2.
\end{aligned}
\end{equation}
Equating the derivative with zero we find for optimality relation $1/(n_1^\mathrm{opt}-1)=\zeta/(n_2^\mathrm{opt}-1)$, which upon substituting into the second equation gives us $n_1^\mathrm{opt}$. Substituting the solution for $n_1^\mathrm{opt}$ back into the optimality relation then gives $n_2^\mathrm{opt}$. The proof is now complete.
\end{proof}

\begin{remark}
\label{rem:tuned_optimal_samples}
Suppose we define $\hat P^{-1}$ in terms of the optimal sample sizes for $\hat P^{-1}_\delta$ and consider what happens as we pass to the shot-noise-limit. Denoting
\begin{equation}
\hat P^{-1}_\mathrm{opt}\coloneqq\frac{1}{\hat X(n_1^\mathrm{opt})-\hat Y(n_2^\mathrm{opt})},
\end{equation}
we see as $\zeta\to 0^+$ two simultaneous events occurring: (1) $n_2^\mathrm{opt}\to 1$, which implies $\hat Y(n_2^\mathrm{opt})\overset{d}{\to}\delta(0)$ and (2) $\sigma_{P+D}\to \sigma_P$ so that $\hat X\overset{d}{\to}\mathcal G(\alpha_1,\alpha_1/\sigma_P^2)$ with $\alpha_1=(n_1^\mathrm{opt}(0)-1)/2$. In conclusion,
\begin{equation}
\lim_{\zeta\to 0^+}\hat P^{-1}_\mathrm{opt}\overset{d}{=}\frac{1}{\hat X(n_1^\mathrm{opt}(0))}
\end{equation}
with the r.h.s. being the shot-noise-limited estimator of $1/\sigma_P^2$ in Section \ref{sec:shot_noise_limited} computed from $n_1^\mathrm{opt}(0)$ observations. Comparing the shot-noise-limited optimal sample sizes for this estimator in (\ref{eq:shot_noise_limited_optimal_samples}) with $n_1^\mathrm{opt}(0)=2/\mathsf{acv}_0+1$ we see that we can force our optimal sample sizes to be exact in the shot-noise-limit by instead using
\begin{equation}
\label{eq:opt_samples_for_hatP_inv_delta}
(n_1^\mathrm{opt},n_2^\mathrm{opt})=\left(\frac{2(1+\zeta)}{\mathsf{acv}_0^2(1-\zeta)^2}+5,\frac{2\zeta(1+\zeta)}{\mathsf{acv}_0^2(1-\zeta)^2}+1\right).
\end{equation}
\end{remark}

\begin{lemma}
\label{lem:Dawson_asymptotic_expansion}
As $|z|\to\infty$
\begin{equation}
\sqrt 2\,\mathcal D(z/\sqrt 2)\sim\frac{1}{z}\sum_{k=0}^\infty(2k-1)!!\frac{1}{z^{2k}}.
\end{equation}
\end{lemma}

\begin{proof}[Proof of Lemma \ref{lem:Dawson_asymptotic_expansion}]
The proof follows from combining the relation
\begin{equation}
\mathcal D(z)=\frac{\sqrt\pi}{2}e^{-z^2}\operatorname{erfi}(z)
\end{equation}
with the asymptotic expansion for $|z|\to\infty$
\begin{equation}
\operatorname{erfi}(z)\sim\operatorname{sgn}(\Im z)i+\frac{1}{\sqrt\pi\,z}e^{z^2}\sum_{k=0}^\infty\frac{(1/2)_k}{z^{2k}},
\end{equation}
and the relation between the Pochhammer symbol and double factorial $(2k-1)!!=2^k(1/2)_k$.
\end{proof}

\begin{proof}[Proof of Theorem \ref{thm:hatPinv_opt_ARB_ACV}]
As $\mathsf{acv}_0\to 0^+$, $n_1^\mathrm{opt}\to\infty$ and $n_2^\mathrm{opt}\to\infty$ so that we have by the central limit theorem $\hat P_\mathrm{opt}\overset{d}{\to}\mathcal N(\mu_{\hat P_\mathrm{opt}},\sigma_{\hat P_\mathrm{opt}}^2)$. Noting that $\mathsf{ACV}\hat P_\mathrm{opt}=\mathsf{ACV}\hat P^{-1}_{\delta,\mathrm{opt}}=\mathsf{acv}_0$ we have after combing (\ref{eq:pseudo_arb}) with (\ref{eq:negative_mgf_normal})
\begin{equation}
\mathsf{ARB}_\mathcal P\hat P^{-1}_\mathrm{opt}\to\left|\frac{\sqrt 2}{\mathsf{acv}_0}\mathcal D\left(\frac{1}{\sqrt 2\,\mathsf{acv}_0}\right)-1\right|.
\end{equation}
Since $\mathsf{acv}_0$ is small, we subsequently have according to Lemma \ref{lem:Dawson_asymptotic_expansion}
\begin{equation}
\mathsf{ARB}_\mathcal P\hat P^{-1}_\mathrm{opt}\sim\mathsf{acv}_0^2\sum_{k=0}^\infty(2k+1)!!\,\mathsf{acv}_0^{2k},
\end{equation}
which leads to the desired result for $\mathsf{ARB}_\mathcal P\hat P^{-1}_\mathrm{opt}$.

By the same line of reasoning we combine (\ref{eq:pseudo_acv2}) and (\ref{eq:negative_mgf_normal})
\begin{equation}
\mathsf{ACV}^2_\mathcal P\hat P^{-1}_\mathrm{opt}\to\frac{\mathsf{ARB}_\mathcal P\hat P^{-1}_\mathrm{opt}}{(\sqrt 2\,\mathcal D(\mathsf{acv}_0^{-1}/\sqrt 2))^2}-1.
\end{equation}
Using standard results for the products and multiplicative inverses of power series we write with the help of Lemma \ref{lem:Dawson_asymptotic_expansion}
\begin{equation}
(\sqrt 2\,\mathcal D(\mathsf{acv}_0^{-1}/\sqrt 2))^{-2}\sim\frac{1}{\mathsf{acv}_0^2}\sum_{k=0}^\infty b_k\,\mathsf{acv}_0^{2k},
\end{equation}
with $b_0=1$, $b_k=-\sum_{i=1}^ka_ib_{k-i}$ and
\begin{equation}
a_i=\sum_{j=0}^i(2(i-j)-1)!!\,(2j-1)!!.
\end{equation}
It follows
\begin{equation}
\mathsf{ACV}^2_\mathcal P\hat P^{-1}_\mathrm{opt}\sim\sum_{k=1}^\infty\left(\sum_{\ell=0}^k(2(k-\ell)+1)!!\, b_\ell\right)\mathsf{acv}_0^{2k},
\end{equation}
which upon combining with $\sqrt{1+ax^2+\mathcal O(x^4)}=1+ax^2/2+\mathcal O(x^4)$ as $x\to 0$ yields the desired asymptotic result for $\mathsf{ACV}_\mathcal P\hat P^{-1}_\mathrm{opt}$.
\end{proof}

\begin{proof}[Proof of Theorem \ref{thm:ACVG_general_approx}]
The proof follows much in the same way as that for Theorem \ref{thm:ACVG_shot_limited_approx}. We write
\begin{equation}
\frac{\mathsf{ACV}G_\delta}{\mathsf{ACV}\hat P^{-1}_\delta}=\left(1+\mathsf{ACV}^2\bar P+(\mathsf{ACV}^2\bar P)(\mathsf{ACV}^{-2}\hat P^{-1}_\delta)\right)^{1/2}.
\end{equation}
Combining this with the asymptotic approximations
\begin{equation}
\mathsf{ACV}^2\bar P=\frac{1}{n_1\sigma_\mathscr D^2}\zeta+\mathcal O(\zeta^2),
\end{equation}
\begin{equation}
\mathsf{ACV}^{-2}\hat P^{-1}_\delta=\frac{n_1-1}{2}+(n_1-1)\zeta+\mathcal O(\zeta^2),
\end{equation}
and $\sqrt{1+ax+\mathcal O(x^2)}=1+ax/2+\mathcal O(x^2)$ as $x\to 0$ gives the desired result.
\end{proof}

\begin{proof}[Proof of Theorem \ref{thm:G_delta_asymptotic_optimal_samples}]
The results of Lemma \ref{lem:barP_optimal_samples} and Proposition \ref{prop:optimal_sample_properties} show that as $\zeta\to 1^-$, the optimal sample sizes for $\hat P^{-1}_\delta$ and $\bar P$ are asymptotically equal and of the form $n_i^\mathrm{opt}\sim C(1-\zeta)^{-2}$. As such, the optimal sample sizes for $G_\delta$ must also be asymptotically equal and of the form $n_i^\mathrm{opt}\sim C_{G_\delta}(1-\zeta)^{-2}$ since
\begin{equation}
\mathsf{ACV}^2G_\delta=\mathsf{ACV}^2\bar P+(\mathsf{ACV}^2\bar P)(\mathsf{ACV}^2\hat P^{-1}_\delta)+\mathsf{ACV}^2\hat P^{-1}_\delta.
\end{equation}
To determine $C_{G_\delta}$ we substitute $n_1=n_2=C_{G_\delta}(1-\zeta)^{-2}$ into the expression for $\mathsf{ACV}^2G_\delta$, pass to the limit, and equate with $\mathsf{acv}_0$ yielding
\begin{equation}
\frac{2}{\sigma_\mathscr D^2}\frac{1}{C_{G_\delta}}+\frac{8}{\sigma_\mathscr D^2}\frac{1}{C_{G_\delta}^2}+4\frac{1}{C_{G_\delta}}=\mathsf{acv}_0.
\end{equation}
The resulting expression is a quadratic equation in $C_{G_\delta}$, the roots of which are real and differing in sign. Taking the positive root then yields the expression for $C_{G_\delta}$.
\end{proof}

\begin{proof}[Proof of Theorem \ref{thm:G_opt_ARB_ACV}]
Since $\mathsf{ARB}_\mathcal PG=\mathsf{ARB}_\mathcal P\hat P^{-1}$, the result for $\mathsf{ARB}_\mathcal PG_\mathrm{opt}$ immediately follows from Theorem \ref{thm:hatPinv_opt_ARB_ACV}. To obtain the result for $\mathsf{ACV}_\mathcal PG_\mathrm{opt}$, combine
\begin{equation}
\mathsf{ACV}_\mathcal P^2G=\mathsf{ACV}_\mathcal P^2\hat P^{-1}+(\mathsf{ACV}_\mathcal P^2\hat P^{-1})(\mathsf{ACV}^2\bar P)+\mathsf{ACV}^2\bar P
\end{equation}
with
\begin{equation}
\begin{aligned}
\mathsf{ACV}^2_\mathcal P\hat P^{-1}_\mathrm{opt} &=\mathsf{acv}_0+6\mathsf{acv}_0^3+\mathcal O(\mathsf{acv}_0^5)\\
\mathsf{ACV}^2_\mathcal P\bar P_\mathrm{opt} &=a\mathsf{acv}_0+b\mathsf{acv}_0^3+\mathcal O(\mathsf{acv}_0^5)
\end{aligned}
\end{equation}
and $\sqrt{1+ax^2+\mathcal O(x^4)}=1+ax^2/2+\mathcal O(x^4)$ as $x\to 0$.
\end{proof}


\section{Backmatter}

\begin{backmatter}

\bmsection{Funding} The authors report no funding for this work.

\bmsection{Acknowledgments} The authors would like thank Nico Schl\"{o}mer for his {\ttfamily{matlab2tikz}} function used to create the figures throughout this work \cite{schlomer_2021}. The authors would also like to thank Paul Enta for pointing us to reference \cite{Schmidt_2017} on square series generating functions.

\bmsection{Disclosures} The authors declare no conflicts of interest.

\bmsection{Data availability} Data underlying the results presented in this paper are not publicly available at this time but may be obtained from the authors upon reasonable request.

\bmsection{Supplemental document}
Code for the {\sc doe} and {\sc coe} procedures is available on figshare (Ref.~\cite{supp_code1,supp_code2}).

\end{backmatter}

\bibliography{sources}

\begin{thebibliography}{10}
\newcommand{\enquote}[1]{``#1''}

\bibitem{janesick_2007}
J.~R. Janesick, \emph{Photon {T}ransfer: $DN\to\lambda$} (SPIE, 2007).

\bibitem{EMVA_1288}
{EMVA 1288 working group}, \enquote{{EMVA} {S}tandard 1288: {S}tandard for
  characterization and presentation of specification data for image sensors and
  cameras, release {A}1.00,}  (2005).

\bibitem{EMVA_1288_4_linear}
{EMVA 1288 working group}, \enquote{{EMVA} {S}tandard 1288: {S}tandard for
  characterization of image sensors and cameras, release 4.0 linear,}  (2021).

\bibitem{Beecken:96}
B.~P. Beecken and E.~R. Fossum, \enquote{Determination of the conversion gain
  and the accuracy of its measurement for detector elements and arrays,}
  {\protect\JournalTitle{Appl. Opt.}} \textbf{35}, 3471--3477 (1996).

\bibitem{janesick_2001}
J.~R. Janesick, \emph{Scientific {C}harge-{C}oupled {D}evices}, Press
  Monographs (SPIE, 2001).

\bibitem{pain_2003}
B.~Pain and B.~R. Hancock, \enquote{Accurate estimation of conversion gain and
  quantum efficiency in {CMOS} imagers,} in \emph{Sensors and Camera Systems
  for Scientific, Industrial, and Digital Photography Applications IV,}
  (2003), pp. 94--103.

\bibitem{janesick_2006}
J.~Janesick, J.~T. Andrews, and T.~Elliott, \enquote{{Fundamental performance
  differences between {CMOS} and {CCD} imagers: {P}art 1},} in \emph{High
  Energy, Optical, and Infrared Detectors for Astronomy II,}  vol. 6276 D.~A.
  Dorn and A.~D. Holland, eds., International Society for Optics and Photonics
  (SPIE, 2006), pp. 208--226.

\bibitem{bohndiek:2008}
S.~E. Bohndiek, A.~Blue, A.~T. Clark, M.~L. Prydderch, R.~Turchetta, G.~J.
  Royle, and R.~D. Speller, \enquote{Comparison of methods for estimating the
  conversion gain of {CMOS} active pixel sensors,} {\protect\JournalTitle{IEEE
  Sensors Journal}} \textbf{8}, 1734--1744 (2008).

\bibitem{starkey_2016}
D.~A. Starkey and E.~R. Fossum, \enquote{Determining conversion gain and read
  noise using a photon-counting histogram method for deep sub-electron read
  noise image sensors,} {\protect\JournalTitle{IEEE Journal of the Electron
  Devices Society}} \textbf{4}, 129--135 (2016).

\bibitem{Hendrickson:17}
A.~J. Hendrickson, \enquote{Centralized inverse-{F}ano distribution for
  controlling conversion gain measurement accuracy of detector elements,}
  {\protect\JournalTitle{J. Opt. Soc. Am. A}} \textbf{34}, 1411--1423 (2017).

\bibitem{hendrickson_2019}
A.~Hendrickson, \enquote{The inverse gamma-difference distribution and its
  first moment in the {C}auchy principal value sense,}
  {\protect\JournalTitle{Statistics {and} Its Interface}} \textbf{12}, 467--478
  (2019).

\bibitem{hendrickson_2021}
A.~Hendrickson, \enquote{A novel approach to photon transfer conversion gain
  estimation,} {\protect\JournalTitle{Cornell University arXiv}}
  \textbf{2106.14958} (2021).

\bibitem{preece_22}
B.~L. Preece and D.~P. Haefner, \enquote{3d noise photon transfer curve,}
  {\protect\JournalTitle{Appl. Opt.}} \textbf{61}, 6202--6212 (2022).

\bibitem{Nakamoto_2022}
K.~Nakamoto and H.~Hotaka, \enquote{Efficient and accurate conversion-gain
  estimation of a photon-counting image sensor based on the maximum likelihood
  estimation,} {\protect\JournalTitle{Opt. Express}} \textbf{30}, 37493--37506
  (2022).

\bibitem{Schmidt_2017}
M.~D. Schmidt, \enquote{Square series generating function transformations,}
  {\protect\JournalTitle{Journal of Inequalities and Special Functions}}
  \textbf{8}, 125--156 (2017).

\bibitem{casella_2002}
G.~Casella and R.~Berger, \emph{Statistical {I}nference}, Duxbury advanced
  series in statistics and decision sciences (Thomson Learning, 2002), 2nd ed.

\bibitem{mathai_1993}
A.~Mathai, \enquote{On noncentral generalized {L}aplacianness of quadratic
  forms in normal variables,} {\protect\JournalTitle{Journal of Multivariate
  Analysis}} \textbf{45}, 239--246 (1993).

\bibitem{krishna_2011}
E.~Krishna and K.~Jose, \enquote{Marshall-{O}lkin generalized asymmetric
  {L}aplace distributions and processes,} {\protect\JournalTitle{Statistica}}
  \textbf{71}, 453--467 (2011).

\bibitem{klar_2015}
B.~Klar, \enquote{A note on gamma difference distributions,}
  {\protect\JournalTitle{Journal of Statistical Computation and Simulation}}
  \textbf{85}, 3708--3715 (2015).

\bibitem{hancova_2022}
M.~Hančová, A.~Gajdoš, and J.~Hanč, \enquote{A practical, effective
  calculation of gamma difference distributions with open data science tools,}
  {\protect\JournalTitle{Journal of Statistical Computation and Simulation}}
  \textbf{92}, 2205--2232 (2022).

\bibitem{Lehmann_1988}
E.~L. Lehmann and J.~P. Shaffer, \enquote{Inverted distributions,}
  {\protect\JournalTitle{The American Statistician}} \textbf{42}, 191--194
  (1988).

\bibitem{PENG_2008}
C.~Y. Peng, \enquote{The first negative moment in the sense of the {C}auchy
  principal value,} {\protect\JournalTitle{Statistics \& Probability Letters}}
  \textbf{78}, 1765--1774 (2008).

\bibitem{PENG_2013}
C.~Y. Peng, \enquote{The first negative moment of skew-t and generalized
  {S}tudent's t-distributions in the principal value sense,}
  {\protect\JournalTitle{Journal of Applied Mathematics}} \textbf{2013} (2013).

\bibitem{fox_1957}
C.~Fox, \enquote{A generalization of the {C}auchy principal value,}
  {\protect\JournalTitle{Canadian Journal of Mathematics}} \textbf{9}, 110--117
  (1957).

\bibitem{Galapon_2016}
E.~A. Galapon, \enquote{The {C}auchy principal value and the {H}adamard finite
  part integral as values of absolutely convergent integrals,}
  {\protect\JournalTitle{Journal of Mathematical Physics}} \textbf{57}, 033502
  (2016).

\bibitem{criscuolo_1997}
G.~Criscuolo, \enquote{A new algorithm for {C}auchy principal value and
  {H}adamard finite-part integrals,} {\protect\JournalTitle{Journal of
  Computational and Applied Mathematics}} \textbf{78}, 255--275 (1997).

\bibitem{QUENOUILLE_1956}
M.~H. Quenouille, \enquote{Notes on bias in estimation,}
  {\protect\JournalTitle{Biometrika}} \textbf{43}, 353--360 (1956).

\bibitem{BARAKAT_1971}
R.~Barakat, \enquote{The derivatives of {D}awson's function,}
  {\protect\JournalTitle{Journal of Quantitative Spectroscopy and Radiative
  Transfer}} \textbf{11}, 1729--1730 (1971).

\bibitem{ccd_sensor}
ON Semiconductor, \emph{KAI-04070: 2048 (H) x 2048 (V) Interline CCD Image
  Sensor} (2013). Rev. 3.

\bibitem{supp_code1}
A.~Hendrickson, D.~P. Haefner, and B.~L. Preece, \enquote{Supplementary
  material 1: {O}n the optimal measurement of conversion gain in the presence
  of dark noise.} figshare (2022).
  {https://opticapublishing.figshare.com/s/e7ce99e8995a5581cdb9}.

\bibitem{supp_code2}
A.~Hendrickson, D.~P. Haefner, and B.~L. Preece, \enquote{Supplementary
  material 2: {O}n the optimal measurement of conversion gain in the presence
  of dark noise.} figshare (2022).
  {https://opticapublishing.figshare.com/s/96ea4554dca4c769a04d}.

\bibitem{LABSPHERE_2017}
Labsphere, \enquote{Integrating sphere theory and applications,}  (2017).

\bibitem{hinkley_1969}
D.~V. Hinkley, \enquote{On the ratio of two correlated normal random
  variables,} {\protect\JournalTitle{Biometrika}} \textbf{56}, 635--639 (1969).

\bibitem{schlomer_2021}
N.~Schl{\"{o}}mer, \enquote{matlab2tikz: {A} script to convert
  {MATLAB}/{O}ctave into {T}ik{Z} figures for easy and consistent inclusion
  into {\LaTeX}.} GitHub. URL: https://github.com/matlab2tikz/matlab2tikz
  (retrieved May 8, 2021).

\end{thebibliography}

\end{document}